\documentclass[10pt]{article}
\usepackage{amsmath}
\usepackage{amssymb, amsthm}
\usepackage{geometry}
\usepackage{graphicx}
\usepackage{hyperref}
\usepackage{indentfirst}
\usepackage{xcolor}
\usepackage{subfigure}
\usepackage{enumerate}
\usepackage{titlesec}
\usepackage{float}
\usepackage{booktabs}
\usepackage{mathrsfs}

\title{Tensor gradient flow for rod-like liquid crystals from molecular model with closure approximation by quasi-entropy}
\author{Yongyong Cai\footnote{Laboratory of Mathematics and Complex Systems and School of Mathematical Sciences, Beijing Normal University, Beijing 100875, China. Email:yongyong.cai@bnu.edu.cn}, Jie Xu\footnote{SKLMS \& NCMIS, Institute of Computational Mathematics and Scientific/Engineering Computing (ICMSEC), Academy of Mathematics and Systems Science (AMSS), Chinese Academy of Sciences, Beijing, China. Email:xujie@lsec.cc.ac.cn}, Haixin Zhang\footnote{Laboratory of Mathematics and Complex Systems and School of Mathematical Sciences, Beijing Normal University, Beijing 100875, China. Email: haixinzhang@mail.bnu.edu.cn}}
\date{}

\newtheorem{theorem}{Theorem}[section]

\newtheorem{proposition}[theorem]{Proposition}
\newtheorem{corollary}[theorem]{Corollary}
\theoremstyle{definition}

\theoremstyle{remark}

\numberwithin{theorem}{section}
\numberwithin{equation}{section}

\newcommand{\md}{\mathrm{d}}
\newcommand{\m}{\mathbf{m}}
\newcommand{\s}{\mathbf{s}}
\newcommand{\n}{\mathbf{n}}
\newcommand{\bx}{\mathbf{x}}
\newcommand{\be}{\mathbf{e}}
\newcommand{\bq}{\mathbf{q}}
\newcommand{\bu}{\mathbf{u}}
\newcommand{\bbi}{\mathbf{i}}
\newcommand{\bbv}{\mathbf{v}}
\newcommand{\tr}{\operatorname{tr}}

\newcommand{\diag}{\operatorname{diag}}
\newcommand{\A}{\mathcal{A}}
\newcommand{\B}{\mathcal{B}}
\newcommand{\Fourcov}{{E}}
\newcommand{\R}{\mathcal{R}}

\newcommand{\e}{\mathbf{e}}
\newcommand{\equaldef}{\overset{\text{def}}{=}}

\newcommand{\Qset}{\mathscr{Q}_{\mathrm{phys}}}

\newcommand{\fM}{{M}}
\newcommand{\entorig}{{\zeta}}

\allowdisplaybreaks[4]

\begin{document}

\maketitle

\begin{abstract}
  In tensor dynamics for liquid crystals derived from molecular models, a common problem is closure approximation. For rod-like molecules, the Bingham closure has proved to outperform other methods because it inherits the gradient flow structure of the molecular model, but is difficult to achieve efficient computations maintaining the gradient flow structure.
  We propose a closure approximation by the quasi-entropy that has been successfully applied to the free energy, based on which we construct the tensor gradient flow. 
  The quasi-entropy closure has the same symmetry properties as the Bingham closure.
  The resulting tensor gradient flow is able to constrain the eigenvalues of the tensor within the physical range, guaranteeing the positive definiteness of the dissipation operator given by the higher-order tensors.
  The quasi-entropy closure is easy to implement since it can be reduced to minimizing an elementary function of three variables.
  As a result, we construct a numerical scheme preserving the eigenvalue constraints and energy dissipation, with the closure approximation decoupled from solving the scheme. 
  Numerical simulations are carried out for the interface between the isotropic and the uniaxial nematic phase, as well as the defect evolutions, where the higher-order tensors indeed make a difference.


\end{abstract}

\section{Introduction}

The structural characteristic of liquid crystals is the local anisotropy
generated by the nonuniform orientational distribution of rigid, typically
rod-like molecules. The local anisotropy has been expermimentally shown to bring about
distinctive configurations both in stationary states and in dynamical
evolutions, for which we refer to~\cite{
    deGennes_Prost_1993,FPTO2020,SmalI2018,STARK2001387,Wang_Zhang_Zhang_2021} and the references therein.

Theoretically, the dynamical models of liquid crystals are given by the
evolutionary equations of order parameters that represent the local anisotropy,
coupled with the Navier-Stokes equation for the velocity. These equations form
an energy dissipative system. Different types of order parameters lead to
models at different levels. For the most commonly observed uniaxial nematic
phase, it is straightforward to use a unit vector as the order parameter. The
corresponding dynamics of the unit vector field is the Ericksen-Leslie
model~\cite{Ericksen_1961,Leslie_1968}, which can be extended to incorporate
flexibility in local anisotropy but is still limited to
uniaxiality~\cite{Calderer2002,Ericksen1991,EL_s2022}.

If one would like to fully characterize the local anisotropy and reach into
molecular information, a natural choice is to consider the density function
w.r.t. both the position and the orientation. The resulting dynamics is a
kinetic equation, first proposed for spatially homogeneous cases in what is known as the
Doi-Onsager theory~\cite{Marrucci_1987} which has been studied
extensively~\cite{kuzuu_doi_1983,kuzuu_doi_1984}. Several attempts were later made to extend the Doi-Onsager theory to inhomogeneous
cases~\cite{2006EWNZPW,Fengjj2000,JGHYHJZPW2008,Wang_2002,WangEWN2002,Inhomogeneous_Yu_2010,Inhomogeneous_Yu_Zhang_2007}. For these models, the
presence of orientational variables in combination with spatial variables brings
significant computational challenges.

To balance the capability of describing local anisotropy with computational
difficulties, tensor models are proposed.
Tensor order parameters are able to distinguish different types of local anisotropy, so that tensor models can describe various phases and their dynamical behaviors. 
The interaction terms in the tensor models can also be derived from the microscopic level. 
Although the results obtained from tensor models of this type usually do not exactly reflect specific molecular architectures, 
they do reveal the changing trends of macroscopic phenomena as the molecular architecture varying, evidenced by several representative cases \cite{Han_Luo_Wang_Zhang_Zhang_2015,Xu_ye_zhang_2018,Xu_Zhang_2017}. 
Based on the above features, the major points of focus in tensor models are the essential mathematical properties and overall structures, rather than the precision of specific terms. 
For rod-like molecules, the order parameter is chosen as a symmetric traceless tensor $Q$.
Dynamic tensor models can either be written down phenomenologically, such as
Beris-Edwards~\cite{Beris_Edwards_1994} and Qian-Sheng models~\cite{QianSheng1998}, or
be derived from molecular
models~\cite{Han_Luo_Wang_Zhang_Zhang_2015}. The latter
type of models possesses clearer physical meanings in terms of molecular shape and
interactions, and are suitable for molecules with more complex shapes~\cite{Xu_Zhang_2017}. However,
a common problem for models of this type is that the evolutionary equations of the
tensor involve higher-order tensors. It is necessary to supplement function
relations between them and the tensor $Q$ to close the system, which is called
the \emph{closure approximation}. Various closure approximations have been
proposed, including some explicit functions~\cite{Cintra_Tucker_1995,Marrucci_1987, Hinch_Leal_1976}, and the Bingham
closure~\cite{Bingham1974,bingham2} defined through the maximum entropy state.
These closure approximations have been compared
thoroughly~\cite{Feng_Chaubal_Leal_1998}, where the Bingham closure is found to
perform better, which is attributed to the fact that the Bingham closure
maintains the energy dissipation.

To comprehend how higher-order tensors affect the energy dissipation, we
concentrate on the cases where the velocity is small. Under the reasonable
approximation that the velocity is taken to be zero, the molecular models
reduce to an evolutionary equation of the density function, while the tensor
models reduce to an evolutionary equation of $Q$. The equation of the density
function can be written as a gradient flow, which ensures its energy
dissipation. However, it is not always the case for the equation of the tensor
$Q$. It turns out that what underlies the fact that the Bingham closure maintains the
energy dissipation is that it manages to preserve the gradient flow structure while
other closure approximations do not. As we will briefly discuss in Sec. \ref{sec:dynamic_Bingham},
when the Bingham closure is adopted, it introduces a singular entropy term in
the free energy that constrains the eigenvalues of $Q$ within $(-1/3,2/3)$. The
dissipation operator is given by a fourth-order tensor that is negative
definite, also guaranteed by the Bingham closure.

Although the Bingham closure maintains the gradient flow structure at the level
of the tensor model, it defines the higher-order tensor as an implicit function
through the maximum-entropy density function (also called the Bingham distribution)
determined by $Q$. It becomes a significant obstacle in computation, since
implementation according to its definition results in high computational cost.
To tackle this problem, fast algorithms for the computation of Bingham closure
itself have been discussed in several
works~\cite{feng2026fastjacobispectralmethods,TYHJJS2021, KAPSPWA2013,Luo_Xu_Zhang_2018,MSWSSM2024,WHLKZPW2008,WSSMSDB2022}.
Nonetheless, they still cannot fill the gap towards solving the evolution
equations efficiently, especially for spatially inhomogeneous cases.

Since it turns out that whether a tensor model works well relies heavily on the
dissipation structure (but not necessarily on the specific form of the Bingham
closure as we have explained above when introducing tensor models),
it is worth seeking alternative simpler approaches for the closure
approximation without breaking such a structure.
In particular, it is better to avoid involving implicit functions through the
density function and integrals on the unit sphere. Actually, this has
successfully been done for the free energy of general liquid crystals (formed by
rigid molecules with arbitrary shape) derived from molecular theory. The
quasi-entropy, defined as the log-determinant of the covariance matrix, has
been proposed for general rigid molecules to substitute the entropy term given
by the maximum entropy density function (that is the Bingham distribution for rod-like molecules)~\cite{XU2022133308}. It is shown that the quasi-entropy preserves the
essential properties of the original entropy term: strict convexity; barrier
function to constrain the covariance matrix positive definite; rotational
invariance; consistency under symmetry reductions. The free energy with the
quasi-entropy is able to capture the underlying physics correctly in
representative cases. Specially, for rod-like molecules, the stationary
points of the bulk energy are shown to be axisymmetric, which are further
classified to describe the isotropic--uniaxial nematic phase transition. On the
other hand, the fact that the quasi-entropy is an elementary function
spontaneously reduces the computational complexity, so that preliminary
discussions have been done on the design of efficient and accurate numerical
methods~\cite{Wang_Xu_2023}, where the models do not contain higher-order
tensors.

Because the free energy is a core ingredient in the dynamic model, we are now
naturally in a position to discuss molecular-theory-based dynamic tensor models
with the free energy including the quasi-entropy. In this paper, we propose the
closure approximation using the quasi-entropy. The basic idea is similar to the
treatment of the free energy, that is, to substitute the original entropy term
with the quasi-entropy wherever it plays a role.
Sec. \ref{sec:dynamic_quasi} is dedicated to the tensor dynamics incorporating
the quasi-entropy in the free energy and the closure approximation,
which we outline below. 

As we have mentioned, the Bingham closure is done by solving the maximum
entropy density function first, then higher-order tensors are calculated. On the
other hand, the original entropy term in the free energy is also calculated by
the maximum entropy density function. But, such a formulation, with the density
function as an intermediate, is difficult for us to carry out the quasi-entropy
substitution. Thus, we first need to figure out an appropriate formulation for
how the closure approximation is directly defined through the entropy term in
the free energy. Actually, we are able to rewrite the Bingham closure as a
constrained minimization problem. For the higher-order tensors appearing in the
dynamic models, they can all be expressed by $Q$ and a fourth-order symmetric
traceless tensor $Q_4$. We consider the maximum entropy density with both $Q$
and $Q_4$ fixed, to obtain a function of these two tensors. Then, the closure
approximation is equivalently defined by minimizing this function of two tensors
with $Q$ fixed. Thus, the solution to this minimization problem expresses $Q_4$
as a function of $Q$. Using the formulation above, we substitute this function
of $Q$ and $Q_4$ with the fourth-order quasi-entropy of $Q$ and $Q_4$, to
establish the closure approximation by the quasi-entropy.

We shall show that the higher-order tensors obtained from the quasi-entropy closure
approximation maintain several essential properties of the Bingham closure, due to
the properties of the quasi-entropy. The strict convexity of the quasi-entropy
guarantees the existence and uniqueness of the closure, with the higher-order
tensor being a continuous function of $Q$. The rotational invariance of the
quasi-entropy leads to the fact that the higher-order tensor rotates together
with the rotation of $Q$. The fourth-order tensor satisfies the same symmetry
as $Q$ as a result of the consistency in symmetry reduction of the
quasi-entropy. The domain of the quasi-entropy guarantees the negative
definiteness of the dissipation operator. With these properties, the
quasi-entropy closure approximation only requires to solve a convex
minimization on an elementary function of three variables, together with tensor
rotations.

In this way, we arrive at a tensor gradient flow derived from molecular theory
with the quasi-entropy determining several core ingredients. The quasi-entropy
gives a singular term in the free energy that is able to constrain the
eigenvalues of $Q$. This enables us to include in the free energy a cubic elastic term bounded from below (involving spatial derivatives). The gradient
flow is indeed energy dissipative, which is a direct consequence of the
negative definiteness of the dissipation operator given by the quasi-entropy
closure approximation.

After the structures of the tensor gradient flow are sufficiently comprehended,
we go on to consider numerical aspects for the gradient flow in Sec. \ref{sec:numerical-method}, with
emphases on preserving these structures. Since the quasi-entropy closure
approximation is extremely easy to implement, we acquire enough flexibility to
deal with other terms. To be as simple as possible, we propose a first-order-in-time scheme constraining the eigenvalues of $Q$ within $(-1/3,2/3)$
and preserving energy dissipation unconditionally, which is also suitable for the
cubic elastic term.
We carry out numerical simulations in 2D, which, to our knowledge, has not
been done before for tensor models requiring closure approximation. Evolutions of
the isotropic--nematic interface and nematic defects are examined, where
differences from gradient flows without incorporating higher-order tensors are
evident.

Although many works on numerical simulations, both for molecular models and
tensor models, have been done before, in this paper we do not attempt to
compare with previous results. The main reason is that the problem settings
and points of focus are quite different between this work and previous works
(and also between those works themselves). To be specific, the previous works
concentrate more on the effect of fluid velocity that is not discussed in this
work. The interaction between rod-like molecules, deduced from the variational
derivative of the free energy, does not have elastic terms but is usually given
by a spatially dependent kernel function. Moreover, the systems that those works
investigate are not actually energy dissipative because the boundary
conditions they impose give rise to energy inputs. Based on the quasi-entropy
closure approximation proposed in this work, we expect in the near future to
extend it to the coupled system with the Navier-Stokes equations (i.e. not taking the
velocity to be zero as an approximation). Comparisons of numerical results will
be done in these forthcoming works. The quasi-entropy closure approximation for
other rigid molecules is also a problem of interests. We shall give further
concluding remarks on the above aspects in Sec. \ref{concl}.

\section{Preliminaries}
For a system of rod-like molecules, the orientation of a single molecule is
represented by a unit vector $ \m \in \mathbb{S}^2$, whose coordinates in the
reference right-handed frame $(\be_1,\be_2,\be_3)$ are denoted by
$m_i,\,i=1,2,3$. Denote by $\R = \m \times \nabla_{\m}$ the rotational gradient
operator on the unit sphere $\mathbb{S}^2$. When written in the coordinates, it
reads $\R _{i} = \epsilon_{ijk}m_{j}{\partial}/{\partial m_k}$, where
$\epsilon$ is the Levi-Civita symbol, and summations on repeated indices are
adopted throughout this work unless stated otherwise. The uniform unit measure
on $\mathbb{S}^2$ is denoted by $\md\m$. For functions $g_1(\m)$, $g_2(\m)$ on
$\mathbb{S}^{2}$, the integration by parts holds as $\int_{\mathbb{S}^2}g_1\R
    _{i}g_2\, \md\m = -\int_{\mathbb{S}^2}g_2\R _{i}g_1\md\m$.

We introduce some notations and basic results for tensors, largely following
\cite{Xu_2020}. An $ n $th-order tensor $ U $ is expressed in terms of the
basis generated by $\be_i$,
\begin{align*}
    U = U_{i_1i_2\dots i_n} \e_{i_1} \otimes \cdots \otimes \e_{i_n}, \quad i_1, \dots, i_n \in \{1, 2, 3\}.
\end{align*}
The dot product between two tensors of the same order is defined as the sum of the products of the corresponding components,
$
    U \cdot W = U_{i_1\dots i_n} W_{i_1\dots i_n} .
$
In particular, the components of a tensor can be given by
$
    U_{i_1\dots i_n} = U \cdot (\e_{i_1} \otimes \dots \otimes\e_{i_n}).
$
Denote the rotation $\mathfrak{t}\in SO(3)$ on a vector $\bq$ as $\mathfrak{t}\circ\bq$.
If expressed in coordinates, $\mathfrak{t}$ is represented by an orthogonal matrix with determinant one and $\mathfrak{t}\circ \bq$ is exactly given by a conventional matrix-vector product.
The rotation can also be imposed on the tensor $U$ by rotating the vectors $\be_i$ while keeping the components unchanged, i.e.,
\[
    \mathfrak t\circ U= U_{i_1\cdots i_n}\,(\mathfrak{t}\circ\e_{i_1})\otimes\cdots\otimes (\mathfrak{t}\circ\mathbf e_{i_n}).
\]
The dot product is rotationally invariant, $ U_1\cdot U_2 = (\mathfrak{t}\circ
    U_1)\cdot(\mathfrak{t}\circ U_2) $. The integral is invariant both under
rotations and inversions, i.e., $ \int_{\mathbb S^2}g(\m) \, \md\m =
    \int_{\mathbb S^2} g(\mathfrak t\circ \m) \, \md\m =\int_{\mathbb S^2}g(-\m)\,
    \md\m $. In other words, the integral is invariant under arbitrary orthogonal
transformations.

A tensor $U$ is \emph{symmetric} if $U_{i_1 \dots i_n} =
    U_{i_{\sigma(1)} \dots i_{\sigma(n)}}$ for any permutation $ \sigma $ of $ \{1,
    \dots, n\} $. For an $n$th-order tensor $ U $, we define its average over all
index permutations as
\begin{equation*}
    {\left(U_{\mathrm{sym}}\right)}_{i_1i_2\dots i_n} = \frac{1}{n!} \sum_{\sigma} U_{i_{\sigma(1)} i_{\sigma(2)} \dots i_{\sigma(n)}},
\end{equation*}
which is a symmetric tensor. We introduce the monomial notation for symmetric tensors generated by basis vectors of an orthonormal frame $(\n_1,\n_2,\n_3)$,
\begin{align*}
    \n_1^{k_1} \n_2^{k_2} \n_3^{k_3} = {\left( \underbrace{\n_1 \otimes \cdots \otimes \n_1}_{k_1} \otimes \underbrace{\n_2 \otimes \cdots \otimes \n_2}_{k_2} \otimes \underbrace{\n_3 \otimes \cdots \otimes \n_3}_{k_3} \right)}_{\mathrm{sym}}.
\end{align*}
In this way, a homogeneous polynomial of $\n_1,\n_2,\n_3$ represents a symmetric tensor.
Since the second-order identity tensor $ \bbi $ satisfies
$
    \bbi = \n_1^2 + \n_2^2 + \n_3^2,
$
we can define $\n_1^{k_1} \n_2^{k_2} \n_3^{k_3}\bbi^l$ in the same manner.

For an $n$th-order symmetric tensor $W$, its trace is defined as an $(n-2)$th-order tensor ${\tr(W)}_{i_1 \dots i_{n-2}}=W_{i_1 \dots i_{n-2} j j}$. If $\tr
    W$ is the zero tensor, $ W $ is called \emph{symmetric traceless}. For any $n$th-order symmetric tensor $U$, there exists a unique symmetric traceless tensor
${(U)}_0=U-{(\bbi\otimes W)}_{\mathrm{sym}}$ (see, for example, Proposition 3.2
in~\cite{Xu_2020}). We call ${(U)}_0$ the symmetric traceless tensor generated
by $U$. For example, $(\m^2)_0=\m^2-{\bbi}/{3}$ is the symmetric traceless
tensor generated by $\m^2$. The $n$th-order symmetric traceless tensors form a
space of dimension $2n+1$. For the second-order symmetric traceless
tensors, an orthogonal basis can be chosen as
\begin{equation}
    \label{eq:the basis of symmetric traceless tensors}
    \s_1 = \n_1^2 - \frac{\bbi}{3},\,\s_2 = \n_2^2 - \n_3^2,\,\s_3 = \n_1\n_2,\,\s_4 = \n_1\n_3,\,\s_5 = \n_2\n_3.
\end{equation}

The orientation distribution at the position $\bx$ is described by a density
function $f\left(\bx,\m\right)>0$ satisfying the normalization condition $
    \int_{\mathbb{S}^2} f\,\md \m = 1$. In tensor models, the order parameter
is chosen as the symmetric traceless tensor generated by the second moment,
defined as $ Q(\bx) = \langle (\m^2)_0\rangle$, where the notation $\langle
    \cdot \rangle$ represents averaging over $\mathbb{S}^2$ w.r.t.  the
density $ f(\bx,\m)$. When necessary, we use the notation
$\langle\cdot\rangle_{f}$ to specify over which density function $f$ the
average is taken.
It is also natural to regard a second-order tensor as a $3\times 3$ matrix, so
that notions for a matrix, such as eigenvalues and eigenvectors, can be
adopted. The definition of $Q$ implies that the eigenvalues $ \lambda(Q) $ lie
within the open interval $ \left(-{1}/{3}, {2}/{3}\right) $. We denote by
$\Qset$ the second-order symmetric traceless tensors satisfying the above
eigenvalue constraints,
\begin{align*}
    \Qset = \left \{ Q : Q_{ij} = Q_{ji}, \, Q_{ii} = 0, \, \lambda(Q) \in \left(-{1}/{3}, {2}/{3}\right) \right \} .
\end{align*}

\section{Tensor dynamics from kinetic equation with the Bingham closure}\label{sec:dynamic_Bingham}
Our starting point is the kinetic equation of the density function $f(\bx,\m,t)$,
which includes both spatial and orientational convection and diffusion terms.
The spatial diffusion term is often omitted because it can be easily recognized
as a higher-order infinitesimal under rescaling. As a result, the equation is
written in the following form,
\begin{equation}
    \label{eq:the molecular model of Smoluchowski equation}
    \frac{\partial f}{\partial t}+\nabla\cdot(\bbv f)=\R \cdot\bigl(\R f+f\R \mu_{\mathrm{r}}\bigr)+\R \cdot\bigl((\m\cdot\nabla)\bbv \times\m f\bigr),
\end{equation}
where $\bbv (\bx,t)$ is the fluid velocity, $\mu_{\mathrm{r}}=\delta F_{\mathrm{r}}/\delta f$ is the interaction potential given by the variational derivative of the interaction energy $F_{\mathrm{r}}$.
For most cases, the interaction energy $F_{\mathrm{r}}$ is a functional of $Q$.
Consequently, the interaction potential $\mu_{\mathrm{r}}$ can be expressed as
\begin{equation}
    \label{eq:mean-field potential}
    \mu_{\mathrm{r}} = \frac{\delta F_{\mathrm{r}}}{\delta f} = \frac{\delta F_{r}}{\delta Q}\cdot\frac{\delta Q}{\delta f}=V_Q\cdot{\left(\m^2\right)}_0,
\end{equation}
where $V_Q=\delta F_{\mathrm{r}}/\delta Q$ depends solely on $Q$.
The velocity $\bbv$ is either prescribed (such as steady shear flows studied extensively \cite{MAMB2005,RIENACKER1999294,Risen2002,Inhomogeneous_Yu_2010,Inhomogeneous_Yu_Zhang_2007}), or obeys a Navier-Stokes equation, which is not written down here since it is beyond the focus of this work.

The dynamic tensor model can be derived from \eqref{eq:the molecular model of
    Smoluchowski equation} by multiplying it with ${\left(\m^2\right)}_0$ and
integrating on $\mathbb{S}^2$, leading to an equation of $Q$,
\begin{equation}
    \label{eq:general equation of tensor model}
    \begin{aligned}
        \frac{\partial Q_{ij}}{\partial t} + \bbv _k\partial_{k} Q_{ij}=
         & -6Q_{ij}-\bigl(2Q_{ik}{(V_Q)}_{kj}+2Q_{jk}{(V_Q)}_{ki}+\frac{4}{3}{(V_Q)}_{ij}-4\langle\m^4\rangle_{ijkl}{(V_Q)}_{kl}\bigr)                                          \\
         & +Q_{ik}\partial_{k}\bbv _j +Q_{jk}\partial_{k}\bbv _i +\frac{1}{3}\bigl(\partial_i\bbv _{j}+\partial_j\bbv _{i}\bigr)-2\langle\m^4\rangle_{ijkl}\partial_{k}\bbv _l.
    \end{aligned}
\end{equation}
Notably, the fourth-order moment $\langle\m^4\rangle$ appears on the right-hand side. It is then necessary to supplement a \emph{closure approximation}, i.e. to express the fourth moment as a function of $Q$. It is worth pointing out that the fourth moment also appears in the Navier-Stokes equation, if it is included as part of the whole system.

It is thus clear that the closure approximation is an essential constituent of
the tensor model deduced from the kinetic equation. As we have mentioned,
various closure approximations have been proposed and the Bingham closure performs
better, which is deemed owing to the maintenance of dissipation structure. Let
us figure out below the dissipation structure and other significant properties
resulting from the Bingham closure.

Since the dissipation structure mainly comes from the diffusion term, we
consider in what follows the case of small velocity and set $\bbv=0$ that
eliminates the convection terms.

Under this assumption, only the diffusion term remains in \eqref{eq:general
    equation of tensor model},
\begin{equation}
    \label{eq:the molecular model of Smoluchowski equation with zero velocity}
    \frac{\partial f}{\partial t}=\R \cdot\bigl(\R f+f\R \mu_{\mathrm{r}}\bigr).
\end{equation}
Noticing that $\R f=f\R (1+\log \! f)$, we rewrite $\R f+f\R \mu_{\mathrm{r}}=f\R \mu$, where $\mu=\delta F/\delta f$ is the variational derivative of the free energy $F$, given by the sum of the interaction energy $F_{\mathrm{r}}$ and the entropy term,
\begin{equation}
    \label{eq:free energy of molecular}
    F=\int \, f\log f\,\md\bx\md\m+F_{\mathrm{r}}[f].
\end{equation}
When $f>0$, \eqref{eq:the molecular model of Smoluchowski equation with zero
    velocity} is a gradient flow of the energy $F$ with the dissipation operator
$\R \cdot\bigl(f\R (\cdot)\bigr)$. Therefore, assuming that the boundary terms
vanish in integration by parts, we deduce the energy dissipation law,
\begin{equation*}
    \frac{\md F}{\md t} = -\int f \left|\R \mu\right|^2\,\md\bx\md\m  \leq 0.
\end{equation*}

Now let us turn to the equation \eqref{eq:general equation of tensor model} of
the tensor $Q$ and assume $\bbv=0$. Define a fourth-order tensor $\fM$,
\begin{equation}
    \label{eq:the structure of dissipative operator M}
    \begin{aligned}
        \fM _{iji'j'}
         & =\langle\R _k{(\m^2)}_0\otimes\R _k{(\m^2)}_0\rangle_{iji'j'} = 4\mathcal{H}{\left(\langle\m^2\otimes(\bbi-\m^2)\rangle\right)}_{iji'j'}             \\
         & = -4\langle\m^4\rangle_{iji'j'} + 4\mathcal{H}{\left(Q\otimes \bbi\right)}_{iji'j'} + \frac{4}{3}\mathcal{H}{\left(\bbi\otimes\bbi\right)}_{iji'j'},
    \end{aligned}
\end{equation}
where ${\mathcal{H}(U)}_{ijkl} = (U_{ikjl}+U_{iljk}+U_{jkil}+U_{jlik})/4$.
For a second-order tensor $W$, denote in short by $MW$ the second-order tensor given by $(MW)_{ij}=M_{iji'j'}W_{i'j'}$.
The equation of the tensor then becomes
\begin{equation}
    \label{eq:tensor model with Bingham closure}
    \frac{\partial Q}{\partial t} = -6Q - \fM V_Q.
\end{equation}
The above form itself does not imply that the equation is a gradient flow.

When the Bingham closure is adopted, \eqref{eq:tensor model with Bingham
    closure} can be rewritten in an equivalent form with the gradient flow
structure, which we explain below. The Bingham closure is defined by solving
the maximum entropy state from $Q$, i.e.,
\begin{equation}\label{eq:the Bingham closure}
    \min_{f} \; \int_{\mathbb{S}^2} f(\m)\log f(\m)\,\md\m
    \qquad \text{s. t.} \qquad
    \int_{\mathbb{S}^2} {(\m^2)}_0 f(\m)\,\md\m = Q,
    \quad
    \int_{\mathbb{S}^2} f(\m)\,\md\m = 1.
\end{equation}
The maximum entropy state is sometimes referred to as the Bingham distribution, denoted by $f_Q$.
The existence and uniqueness of $f_Q$ have been established (see, e.g.,~\cite{Ball_Majumdar_2010,LSRWWZPW2015,Xu_2020,Xu_ye_zhang_2018}).
Together with symmetry properties of $f_Q$, the results are summarized in the following.

\begin{proposition}\label{pro:the existence and uniqueness of the Bingham closure}
    For any $Q\in \Qset$, \eqref{eq:the Bingham closure} admits a unique solution of the form
    \begin{equation}\label{eq:the maximum entropy state}
        f_Q(\m)= \frac{1}{Z_Q}\exp\bigl(B_Q \cdot {(\m^2)}_0\bigr),
        \qquad
        Z_Q = \int_{\mathbb{S}^2}
        \exp\bigl(B_Q \cdot {(\m^2)}_0\bigr)\,\md\m ,
    \end{equation}
    where $B_Q$ is a second-order symmetric traceless tensor. The density $f_Q$ satisfies the following properties:
    \begin{enumerate}
        \item $f_{Q}(\m)=f_{Q}(-\m)$, $f_{\mathfrak{t}\circ Q}(\m) = f_{Q}(\mathfrak{t}^{-1}\circ \m)$ for $\mathfrak{t}\in SO(3)$.
        \item Suppose $Q$ is expressed using its eigenvectors as
              \begin{equation}
                  \label{eq:the spectral decomposition of Q}
                  Q = s\s_1+b\s_2 = s\left(\n_1^2-\bbi/3\right) + b\left(\n_2^2-\n_3^2\right), \quad (\n_1,\n_2,\n_3)\in SO(3),
              \end{equation}
              where $s$ and $b$ are scalars.
              Then, $B_Q$ also has the same eigenframe $(\n_1,\n_2,\n_3)$, and $f_{Q}(\m) = f_{Q}(\m - 2(\m\cdot\n_{i})\n_{i}),\, i=1,2,3$.
              When $Q$ is uniaxial, i.e. $b=0$, it holds that $B_Q$ is also uniaxial, and $f_{Q}(\m) = \widetilde{f}_{Q}\bigl({\left(\m\cdot\n_1\right)}^2\bigr)$ is a function of $(\m\cdot\n_1)^2$.
    \end{enumerate}
\end{proposition}

\begin{proof}
    For the existence and uniqueness, we refer to the proofs given in the literature \cite{Ball_Majumdar_2010,LSRWWZPW2015,Xu_2020,Xu_ye_zhang_2018}.

    For the first property, $f_{Q}(\m)=f_{Q}(-\m)$ follows immediately from
    \eqref{eq:the maximum entropy state}. The rotational invariance of integral and
    dot product yields
    \begin{equation*}
        \mathfrak{t}\circ Q = \int_{\mathbb{S}^2} {(\mathfrak{t}\circ \m)}^2_0 f_{Q}(\m)\md\m = \int_{\mathbb{S}^2} {(\m^2)}_0 f_{Q}(\mathfrak{t}^{-1}\circ\m)\md\m.
    \end{equation*}
    Note that $f_{Q}(\mathfrak{t}^{-1}\circ\m)=\exp\big((\mathfrak{t}\circ B_Q)\cdot(\m)^2_0\big)/Z_Q$ also takes the form \eqref{eq:the maximum entropy state}.
    Together with the uniqueness, we conclude that $f_{\mathfrak{t}\circ Q}(\m) = f_{Q}(\mathfrak{t}^{-1}\circ \m)$.

    For the second property, we consider the following minimization problem,
    \begin{align}\label{minprob_diag}
        \min_{f} \; \int_{\mathbb{S}^2} f(\m)\log f(\m)\,\md\m
        \quad\text{s. t.}
        \quad
        \int_{\mathbb{S}^2} {(\m^2)_0\cdot\s_i} f(\m)\,\md\m = Q\cdot\s_i,\,i=1,2,\quad  \int_{\mathbb{S}^2} f(\m)\,\md\m = 1.
    \end{align}
    The solution takes the form,
    \begin{equation}\label{density_diag}
        \widehat{f}(\m) = \frac{1}{\widehat{Z}}\exp\bigl(\eta_{1}\left(\m^2\right)_{0} \cdot (\n_1^2-\bbi/3) + \eta_{2}\left(\m^2\right)_{0} \cdot (\n_2^2-\n_3^2)\bigr).
    \end{equation}
    Writing
    \begin{equation}\label{coor_eig}
        \m=(\m\cdot\n_1)\n_1+(\m\cdot\n_2)\n_2+(\m\cdot\n_3)\n_3,
    \end{equation}
    we have
    \begin{equation}\nonumber
        \langle \m^2\rangle_{\widehat{f}}=\langle(\m\cdot\n_i)(\m\cdot\n_j)\rangle_{\widehat{f}}\n_i\otimes\n_j.
    \end{equation}
    It follows from \eqref{density_diag} that, for the orthogonal transformation $\m\mapsto \m - 2(\m\cdot\n_{3})\n_{3}$, $\widehat{f}(\m) = \widehat{f}(\m - 2(\m\cdot\n_{3})\n_{3})=\widehat{f}((\m\cdot\n_1)\n_1+(\m\cdot\n_2)\n_2-(\m\cdot\n_3)\n_3)$.
    Therefore, when $i\ne j$, $\langle(\m\cdot\n_i)(\m\cdot\n_j)\rangle_{\widehat{f}}=0$.
    The fact that $Q$ is traceless leads to $\langle(\m^2)_0\rangle_{\widehat{f}} = Q$.
    Again, the uniqueness implies $f_Q=\widehat{f}$.
    The case where $Q$ is uniaxial can be similarly shown by eliminating the constraint on $Q\cdot\s_2$ in \eqref{minprob_diag}.
    It makes $\eta_2=0$ in the solution \eqref{density_diag}.
    As a result, $\widehat{f}$ does not change under the orthogonal transformation $\m\mapsto (\m\cdot\n_1)\n_1+(\m\cdot\n_3)\n_2+(\m\cdot\n_2)\n_3$.
    Thus, we deduce that $\langle(\m\cdot\n_2)^2\rangle_{\widehat{f}}=\langle(\m\cdot\n_3)^2\rangle_{\widehat{f}}$.
    Together with the traceless condition, we obtain $\langle(\m^2)_0\rangle_{\widehat{f}}=s(\n_1^2)_0$, and finally arrive at $f_Q=\widehat{f}=\widetilde{f}_Q\bigl((\m\cdot\n_1)^2\bigr)$ for some \(\widetilde{f}_Q\).
\end{proof}

The fourth-order tensor $\fM $ is then calculated from $f_Q$. Since $f_Q$ is
uniquely determined by $Q$, the tensor $\fM$ is expressed as a function of $Q$.
We denote by $\fM ^{\mathrm{Bin}}$ the tensor $\fM$ determined from the Bingham
closure.

Armed with the density $f_Q$, we define its entropy $\entorig$, which is a
function of $Q$,
\begin{equation}
    \label{eq:the free energy of the tensor model with Bingham closure}
    \entorig(Q) = \int_{\mathbb{S}^2} f_Q(\m)\log f_Q(\m)\md\m = B_Q \cdot Q - \log Z_Q.
\end{equation}

\begin{proposition}\label{th:the Bingham closure of gradient flow}
    The tensor model \eqref{eq:tensor model with Bingham closure} together with the Bingham closure can be written as a gradient flow,
    \begin{equation}
        \label{eq:gradient flow of Bingham closure}
        \begin{aligned}
             & \frac{\partial Q}{\partial t}=-\fM ^{\mathrm{Bin}}\mu_Q,\quad \mu_Q=\frac{\delta F}{\delta Q},\quad F = \int  \entorig(Q)\md\bx + F_{\mathrm{r}}[Q].
        \end{aligned}
    \end{equation}
    When $\lambda(Q)\in (-1/3,2/3)$, the fourth-order tensor $\fM^{\mathrm{Bin}}$ is positive definite in the sense that $W\cdot \fM^{\mathrm{Bin}} W> 0$ for any nonzero second-order symmetric traceless tensor
    $W$. Therefore, the following energy dissipation law holds,
    \begin{equation}
        \frac{\md F}{\md t}=-\int \mu_Q\cdot \fM ^{\mathrm{Bin}}\mu_Q\,\md\bx\leq 0.
    \end{equation}
\end{proposition}

\begin{proof}
    From the definition of $F$ and $Q=(1/Z_Q)\partial Z_Q/\partial B_Q=\partial (\log Z_Q)/\partial B_Q$, we obtain
    \begin{equation*}
        \mu_Q = \frac{\partial \entorig}{\partial Q} + V_Q = B_Q+\frac{\partial B_Q}{\partial Q}\cdot Q -\frac{\partial \log Z_Q}{\partial B_Q}\cdot\frac{\partial B_Q}{\partial Q} + V_Q = B_Q + V_Q.
    \end{equation*}
    From the fact that $\R \cdot\R (\m^2)_0=-6(\m^2)_0$, we deduce from integration by parts and an identity similar to the equation of $f(\bx,\m)$, i.e. $\R f_Q=f_Q\R (\log f_Q)=f_Q\R (B_Q\cdot (\m^2)_0)$, that
    \begin{equation}
        \begin{aligned}
            6Q & =-\int_{\mathbb{S}^2}f_{Q}\R _i(\R _i(\m^2)_0)\,\md\m
            =\int_{\mathbb{S}^2}\R _if_{Q}\R _i(\m^2)_0 \,\md\m                                                               \\
               & =\int_{\mathbb{S}^2}f_{Q}\left( \R _i\left[{\left(\m^2\right)}_0\cdot B_Q\right]\right)\R _i(\m^2)_0 \,\md\m
            =\fM ^{\mathrm{Bin}} B_Q.
        \end{aligned}
    \end{equation}
    Therefore,
    \begin{equation*}
        -6Q - \fM ^{\mathrm{Bin}} V_Q = -\fM ^{\mathrm{Bin}} (B_Q + V_Q) = -\fM ^{\mathrm{Bin}} \mu_Q.
    \end{equation*}
    This shows that the tensor model~\eqref{eq:general equation of tensor model},
    closed using the Bingham closure, can be reformulated as the gradient flow~\eqref{eq:gradient flow of Bingham closure}.
    Using $\fM^{\mathrm{Bin}}=\langle\R _i{(\m^2)}_0\otimes\R _i{(\m^2)}_0\rangle$, we deduce that
    \[
        \begin{aligned}
            W\cdot \fM ^{\mathrm{Bin}}W
            = \int_{\mathbb{S}^2}
            f_Q(\m) \sum_{i=1}^3(W\cdot\R _i(\m^2)_0)^2\md\m\ge 0.
        \end{aligned}
    \]
    Since $f_Q>0$, the equality holds only if $W\cdot \R _{i}(\m^2)_0=0, \, i=1,2,3$ for any $\m\in
        \mathbb{S}^2$, yielding $W=0$. The energy dissipation law follows directly from
    the positive definiteness of $\fM^{\mathrm{Bin}}$.
\end{proof}

For the closure approximation to be successfully carried out, one requires
$\lambda(Q)\in (-1/3,2/3)$ in the equation of $Q$. This can actually be
constrained by the term $\entorig(Q)$ in the free energy, which is a barrier
function. The properties of $\entorig(Q)$ are also discussed in the literature
(cf.~\cite{Ball_Majumdar_2010,LSRWWZPW2015,XU2022133308}), which are summarized below.

\begin{proposition}\label{pro:the properties of the original entropy}
    \begin{enumerate}
        \item $\entorig(Q)$ is strictly convex on $\Qset$ with the unique minimizer $Q = 0$.
        \item $\entorig$ is rotationally invariant: $\entorig(Q) = \entorig(\mathfrak{t}\circ Q)$, for any $\mathfrak{t}\in SO(3)$.
        \item $\entorig$ is a barrier function in the sense that $\lim_{\lambda(Q)\to (-1/3)^+,{(2/3)}^-}\entorig(Q)=+\infty$.
    \end{enumerate}
\end{proposition}

\begin{proof}
    The strict convexity follows from that of $f\log f$, and $Q=0$ corresponds to the case that $f_Q$ is a constant. The rotation invariance is straightforward from that of $f_Q$. The third property calls for quite a few efforts, for which we refer to \cite{Ball_Majumdar_2010,LSRWWZPW2015}.
\end{proof}

The fourth-order tensor $\fM ^{\mathrm{Bin}}$ also satisfies a few symmetry properties resulting from those of $f_Q$.

\begin{proposition}\label{pro:the property of the fourth-order tensor Q4}
    $\fM ^{\mathrm{Bin}}$ rotates along with $Q$ as
    $ \fM ^{\mathrm{Bin}}\!\left(\mathfrak{t}\circ Q\right) = \mathfrak{t}\circ \fM ^{\mathrm{Bin}}(Q)$ for $\mathfrak{t}\in SO(3)$.
    Moreover, when $Q$ is written by its eigenframe as \eqref{eq:the spectral decomposition of Q}, $\fM ^{\mathrm{Bin}}$ takes the form
    \begin{equation}\label{eq:eigen-M}
        \widetilde{M}_{ii}\,\s_i\otimes\s_i + \widetilde{M}_{12}\,(\s_1\otimes\s_2+\s_2\otimes\s_1),
    \end{equation}
    where $ {\{ \s_i \}}_{i=1}^5$ are defined in~\eqref{eq:the basis of symmetric traceless tensors}. When $Q$ is uniaxial ($b=0$ in \eqref{eq:the spectral decomposition of Q}), $\fM ^{\mathrm{Bin}}$ takes the form
    \begin{equation}
        \label{eq:uniaxial-MBin-n1}
        \begin{aligned}
            \alpha_1 \n_1^4 + \alpha_2\mathcal{H}(\n_1^2\otimes\bbi)+\alpha_3\n_1^2\bbi+\alpha_4\mathcal{H}(\bbi\otimes\bbi)+\alpha_5{\bbi}^2.
        \end{aligned}
    \end{equation}
\end{proposition}

\begin{proof}
    $ \fM ^{\mathrm{Bin}}\!\left(\mathfrak{t}\circ Q\right) = \mathfrak{t}\circ \fM ^{\mathrm{Bin}}(Q)$ is a straightforward result of $f_{\mathfrak{t}\circ Q}(\m)=f_Q(\mathfrak{t}^{-1}\circ\m)$.

    From the definition \eqref{eq:the structure of dissipative operator M} of $\fM
        ^{\mathrm{Bin}}$, it is not difficult to verify that
    \[
        \fM ^{\mathrm{Bin}}_{iji{'}i{'}}=\fM ^{\mathrm{Bin}}_{iii{'}j{'}}=0,
        \qquad
        \fM ^{\mathrm{Bin}}_{iji{'}j{'}}=\fM ^{\mathrm{Bin}}_{jii{'}j{'}}
        =\fM ^{\mathrm{Bin}}_{ijj{'}i{'}}=\fM ^{\mathrm{Bin}}_{i{'}j{'}ij}.
    \]
    Therefore, $\fM ^{\mathrm{Bin}}$ can be expressed as
    \begin{equation}\label{eq:basis expansion of M}
        \fM ^{\mathrm{Bin}}
        = \widetilde{M}_{ij}\,\s_i\otimes\s_j,\qquad \widetilde{M}_{ij}=\widetilde{M}_{ji},
        \qquad i,j=1,\dots,5.
    \end{equation}
    Substituting \eqref{coor_eig} into $\fM ^{\mathrm{Bin}}=\langle\mathcal{H}(\m^2\otimes(\bbi-\m^2))\rangle_{f_Q}$ and using $f_Q(\m)=f_Q\big(\m-2(\m\cdot\n_i)\n_i\big)$ for $i=1,2,3$, we conclude that each $\n_i$ must appear even times in $\s_i\otimes\s_j$, which implies \eqref{eq:eigen-M}.

    For the uniaxial case $b=0$, recall that $f_{Q}(\m) =
        \widetilde{f}_{Q}((\m\cdot\n_1)^2)$, which implies that
    $f_Q(\m)=f_Q\big((\m\cdot\n_1)\n_1+\sqrt{1-{(\m\cdot\n_1)}^2}(\n_2\cos\varphi+\n_3\sin\varphi)\big)$
    for any $\varphi\in[0,2\pi)$. Thus, we have
    \begin{equation}
        \begin{aligned}
            \langle\m^4\rangle
             & = \int_{\mathbb{S}^2} \widetilde{f}_{Q}((\m\cdot\n_1)^2)\frac{1}{2\pi}\int_{0}^{2\pi} \bigl((\m\cdot\n_1)\n_1+\sqrt{1-{(\m\cdot\n_1)}^2}(\n_2\cos\varphi+\n_3\sin\varphi)\bigr)^4 \md\varphi \md\m.
        \end{aligned}
    \end{equation}
    It can be written as a linear combination of $\int_{0}^{2\pi}\n_1^{k}\left(\n_2\cos\varphi+\n_3\sin\varphi\right)^{4-k}\md\varphi$ for $0\le k\le 4$.
    Straightforward calculations lead to
    \begin{equation}
        \begin{aligned}
            \langle\m^4\rangle
             & = \beta_1 \n_1^4 + \beta_2 \left(\n_1^2\n_2^2 + \n_1^2\n_3^2\right) + \beta_3 \left( \n_2^4 + \n_3^4 + 2\n_2^2\n_3^2 \right) \\
             & = (\beta_1 - \beta_2 + \beta_3) \n_1^4 + (\beta_2 - 2\beta_3) \n_1^2\bbi + \beta_3 \bbi^2,
        \end{aligned}
    \end{equation}
    where we have utilized $\bbi=\n_1^2+\n_2^2+\n_3^2$. Taking the above equality and $Q=s(\n_1^2-\bbi/3)$ into the definition \eqref{eq:the structure of dissipative operator M}, we obtain the form~\eqref{eq:uniaxial-MBin-n1}.
\end{proof}

To summarize, the Bingham closure is carried out by solving the maximum entropy
state $f_Q$ from $Q$, followed by calculating higher-order tensors using $f_Q$.
The gradient flow structure is recognized by rewriting $6Q=\fM
    ^{\mathrm{Bin}}(\partial \zeta/\partial Q)$. The singular term $\zeta(Q)$
constrains $\lambda(Q)\in (-1/3,2/3)$, so that $f_Q$ can be uniquely
solved. Furthermore, $\fM ^{\mathrm{Bin}}$ is positive definite as a result of
$f_Q>0$, and enjoys several symmetry properties inherited from $f_Q$.

On the other hand, it is important that these structures and properties are maintained in computations, which brings about difficulties. 
If we carry out the Bingham closure through $f_Q$, we need to handle implicit functions involving integrals on $\mathbb{S}^2$.
Naive implementations of these functions lead to high computational cost~\cite{JGHYHJZPW2008,Luo_Xu_Zhang_2018}. 
To this end, several fast algorithms have been developed for the Bingham closure using various approximation formulae \cite{TYHJJS2021,KAPSPWA2013,Luo_Xu_Zhang_2018,WHLKZPW2008,WSSMSDB2022} or neural networks if high accuracy is not needed \cite{feng2026fastjacobispectralmethods,MSWSSM2024,SBMNN2026,shi2026molecularmodeltensormodel}, including the mappings between $Q$ and $B_Q$~\cite{feng2026fastjacobispectralmethods,WHLKZPW2008}, as well as the computation of $\zeta(Q)$, $Z_Q$~\cite{Luo_Xu_Zhang_2018} and the fourth-order tensor $\fM ^{\mathrm{Bin}}$ from $Q$ or $B_Q$~\cite{Grooso_Massimiliano_Closure_fast_algramias,TYHJJS2021,KAPSPWA2013,WSSMSDB2022}. 
However, there is still a large gap towards efficient numerical methods that preserve the desired structures and properties. 
One may choose to change the independent variables to $B_Q$ from which $Q$ and $\fM ^{\mathrm{Bin}}$ can be calculated~\cite{TYHJJS2021}. 
This gives rise to much more complicated terms for the functional derivatives of the interaction energy $F_{\mathrm{r}}$. 
An alternative way is to utilize the form \eqref{eq:tensor model with Bingham closure} that only requires to compute $\fM ^{\mathrm{Bin}}$ from $Q$.
But it literally disregards the gradient flow structure and the eigenvalue constraints, so that they are difficult to be guaranteed in the schemes. 
If schemes are built based on \eqref{eq:gradient flow of Bingham closure} with fast algorithms from $Q$ to $B_Q$, the main obstacle lies in the singularities for $\lambda(Q)$ close to the boundary of $\Qset$, which restricts sufficient accuracy, and in turn impair the convexity of $\zeta(Q)$ up to certain precision required for the solvability of the schemes.

The implementations of the Bingham closure bring limitations in the numerical simulations for tensor dynamics derived from the molecular models. 
Indeed, numerical simulations have been carried out only for 1D, or 2D with the orientation restrained on the unit circle $\mathbb{S}^1$~\cite{JGHYHJZPW2008,Inhomogeneous_Yu_Zhang_2007}. 
The difficulties become worse if rigid molecules of complex shapes are considered, where more order parameters and higher-order tensors are present~\cite{xu_2020_general}.

As we have mentioned at the beginning of this article, for tensor models the emphasis is the essential properties and structures, instead of the precisions of specific terms, especially the Bingham closure on which many works have put efforts. 
This can also be recognized from the interaction energy $F_{\mathrm{r}}$, usually given by a few terms in the gradient expansion that are only able to roughly depict the molecular information. 
Thus, given that the essential properties and structures are maintained, it is desirable if we are able to seek alternative ways to treat the terms involving $f_Q$ (including $B_Q$, $\zeta(Q)$, and $\fM$). 

Actually, in the free energy, the quasi-entropy has been proposed as a substitute for $\zeta(Q)$.
In the following section, we introduce the quasi-entropy and discuss how to use it in the closure approximation, so that we construct a tensor dynamics based on the quasi-entropy.

\section{Tensor dynamics with the quasi-entropy}\label{sec:dynamic_quasi}

The quasi-entropy \cite{XU2022133308} is a class of elementary functions of tensors averaged on $SO(3)$. 
It is proposed for general choice of tensor order parameters, to act as the entropy term in the free energy. 
Generally, the quasi-entropy is defined as the log-determinant of the covariance matrix for the tensors up to certain even order. 
It is shown that the quasi-entropy possesses the essential properties of the entropy defined from the maximum entropy state, including strict convexity, enforcing positive definiteness of the covariance matrix, invariance under rotations, and consistency under symmetry reduction. 
Moreover, for various shapes of molecules, the free energy incorporating the quasi-entropy is able to capture the underlying physics correctly. 
We shall see shortly that the above properties are also crucial for the closure approximation.

Below, we start from briefing the results for the case of one tensor $Q$.
Then, we discuss the closure approximation using the quasi-entropy, which is built on a reformulation of the closure approximation as a minimization problem w.r.t. a fourth-order tensor. 

\subsection{Quasi-entropy in the free energy}

Recall that for the free energy given in \eqref{eq:gradient flow of Bingham closure}, the entropy term $\zeta(Q)$ is defined in \eqref{eq:the free energy of the tensor model with Bingham closure}.
For the tensor $Q$, we use the second-order quasi-entropy $\Xi_2(Q)$ to substitute $\zeta(Q)$, which reads 
\begin{align}
    \label{eq:second-order quasi-entropy}
    \Xi_2(Q)=-\log\det\left(Q+\frac{\bbi}{3}\right)-2\log\det\left(\frac{\bbi}{3}-\frac{Q}{2}\right).
\end{align}
The domain of $\Xi_2$ consists of all $Q$ such that $Q+\bbi/3$ and $\bbi/3-Q/2$ are positive definite, which is exactly $\Qset$.
The quasi-entropy $\Xi_2$ possesses the essential properties of $\zeta$.
\begin{proposition}\label{propsition:quasi-entropy xi2}
    \begin{enumerate}
        \item $\Xi_2(Q)$ is strictly convex on $\Qset$ with the unique minimizer $Q=0$.
        \item $\Xi_2$ is rotationally invariant:
              $\Xi_2(Q)=\Xi_2\left(\mathfrak{t}\circ Q\right)$, for any \(\mathfrak{t}\in SO(3)\).
        \item $\Xi_2$ is a barrier function constraining $Q\in \Qset$ by $\lim_{\lambda(Q)\to (-1/3)^+,{(2/3)}^-}\entorig(Q)=+\infty$.
    \end{enumerate}
\end{proposition}
For the proof, see Proposition 2.1 in \cite{Wang_Xu_2023} (also Theorem 4.8 in~\cite{XU2022133308} for general cases). 

To further illustrate the feature of $\Xi_2$, we consider the bulk energy containing a quadratic term,
\begin{equation}\label{eq:bulk energy}
  \Xi_2(Q)-\frac{1}{2}c_{02}\lvert Q\rvert^2.
\end{equation}
When $\zeta(Q)$ takes the place of $\Xi_2(Q)$, we recover exactly the Maier-Saupe energy.
It is known that the stationary points of the Maier-Saupe energy are uniaxial.
The same results hold when using the quasi-entropy. 
\begin{proposition}\label{pro:stationary proposition is uniaxial or isotropic}
    The stationary points of \eqref{eq:bulk energy} take the form $Q=s(\n^2-\bbi/3)$ where $\n$ is a unit vector.
\end{proposition}
Based on the uniaxial expression, the stationary points have been completely classified \cite{XU2022133308}, for which two critical values of $c_{02}$ describe the first-order isotropic--nematic phase transition.
When the elastic energy (terms with spatial derivatives) is considered, the defect morphology has also been validated in an $L^2$ gradient flow without including the higher-order tensors \cite{Wang_Xu_2023}. 
Since the free energy is a core ingredient of the dynamic model, it is natural to discuss the treatment of higher-order tensors using quasi-entropy in a consistent way with the free energy. 

\subsection{Closure approximation by the quasi-entropy}
It is clear from the discussions in the previous section that the closure approximation is also determined through entropy.
However, it is done by solving the maximum entropy state,
making a substitution like the free energy not quite evident.
In what follows, we reformulate the Bingham closure as a minimization problem w.r.t. an entropy function of $Q$ and $\fM$.

By far, we have not gone into a basic problem: the independent variables of the tensor $\fM$, which is unimportant with a density function in hand. 
Actually, the independent variables of a tensor can be given by symmetric traceless tensors \cite{Xu_2020}.
For $\fM $, defined in~\eqref{eq:the structure of dissipative operator M}, we express it by $Q$ and a fourth-order symmetric traceless tensor $Q_4$, 
\begin{equation}
    \label{eq:the dissipative operator M}
    \fM  = -4Q_4 - \frac{4}{7}\mathcal{A}(Q) - \frac{2}{15}\Fourcov ,
\end{equation}
where
\begin{equation}
    \label{eq:symtrls-M}
    \begin{aligned}
         & Q_4 = \langle {(\m^4)}_0 \rangle = \langle \m^4 - \frac{6}{7} \m^2 \bbi + \frac{3}{35} \bbi^2 \rangle,                                                               \\
         & \Fourcov _{ijkl} = 2\delta_{ij} \delta_{kl} - 3\delta_{ik} \delta_{jl} - 3\delta_{il} \delta_{jk},                                                                   \\
         & \mathcal{A}{(Q)}_{ijkl} = Q_{ij} \delta_{kl} + Q_{kl} \delta_{ij} - \frac{3}{4} (Q_{ik} \delta_{jl} + Q_{il} \delta_{jk} + Q_{jk} \delta_{il} + Q_{jl} \delta_{ik}).
    \end{aligned}
\end{equation}
We are now ready to reformulate the Bingham closure. Define
\begin{equation}
    \label{eq:the equvialent form of Bingham closure Step 1}
    \begin{aligned}
         & \widetilde{\zeta}(Q,Q_4)=\min \int_{{\mathbb{S}}^2} f \log f\,\md\m \quad  \text{s. t. }  \quad \int_{{\mathbb{S}}^2} f \,\md\m=1, \quad \langle {\left(\m^2\right)}_0 \rangle = Q, \quad \langle{\bigl(\m^4\bigr)}_0 \rangle = Q_4.
    \end{aligned}
\end{equation}

\begin{proposition}
    The function $\widetilde{\zeta}(Q,Q_4)$ is well defined for $(Q,Q_4)$ such that they are averages of $\big((\m^2)_0,(\m^4)_0\big)$ under certain $0\le f(\m)<+\infty$.
    The Bingham closure $\fM ^{\mathrm{Bin}}$ is given by \eqref{eq:the dissipative operator M} where $Q_4$ solves the minimization problem, i.e.
    \begin{equation}
        \label{eq:the minimization problem of the Bingham closure for Q_4}
        \langle(\m^4)_0\rangle_{f_Q} = \underset{Q_4}{\arg\min} \,\widetilde{\zeta}(Q,Q_4).
    \end{equation}
\end{proposition}
\begin{proof}
    For the minimization problem given in \eqref{eq:the equvialent form of Bingham closure Step 1}, there exists a unique solution $f(\m)$ \cite[Theorem 5.1]{Xu_2020}, so that $\widetilde{\zeta}(Q,Q_4)$ is well-defined.
    Then, \eqref{eq:the minimization problem of the Bingham closure for Q_4} follows immediately by noticing the existence and uniqueness result stated in Proposition \ref{pro:the existence and uniqueness of the Bingham closure}, because $f_Q$ is minimizes the entropy term with fewer constraints in \eqref{eq:the Bingham closure} than in \eqref{eq:the equvialent form of Bingham closure Step 1}.
\end{proof}

The convenience of the formulation
\eqref{eq:the equvialent form of Bingham closure Step 1}
is that there is an entropy function $\widetilde{\zeta}(Q,Q_4)$.
In a similar way of dealing with the free energy, we now substitute $\widetilde{\zeta}(Q,Q_4)$ with the quasi-entropy.
Since the entropy function depends on the fourth-order tensor $Q_4$, we use the fourth-order quasi-entropy, denoted by $\Xi_4$, in the substitution. 
The function $\Xi_4$ can be found in previous works for cases with more tensors for other symmetries \cite{XU2022133308,li2023}. 
A more straightforward version is given by \cite[(5.14)-(5.18)]{XU2022133308}, for which we set $M^4_1=0$ by symmetry reduction and arrive at the expressions below. 

Let $U^{[j]}$ be a $j$th-order tensor and $(\n_1,\n_2,\n_3)\in SO(3)$, and define
vectors and matrices as follows,
\begin{align*}
    &\Phi_{2}{\left(U^{[2]}\right)}_{i}=U^{[2]}\cdot\s_i,  \quad
    \Psi_{2}{\left(U^{[2]}\right)}_{ij}=U^{[2]}\cdot\n_i\otimes\n_j, \\
    &\Psi_{3}{\left(U^{[3]}\right)}_{ij}=U^{[3]}\cdot\n_i\otimes\s_j, \quad
    \Psi_{4}{\left(U^{[4]}\right)}_{ij}=U^{[4]}\cdot\s_i\otimes\s_j,
\end{align*}
where \(\s_{j},j=1,\dots,5\), are defined in~\eqref{eq:the basis of symmetric traceless tensors}. Define
\begin{equation}
    \label{eq:the three matrices of Psi_4}
    \begin{aligned}
         & C_1=Q_4-\frac{4}{21}\A(Q)-\frac{1}{45}\Fourcov,\  C_2=\frac{1}{8}Q_4+\frac{1}{7}\A(Q)-\frac{1}{60}\Fourcov, \
         & C_3=-\frac{1}{2}Q_4-\frac{1}{14}\A(Q)-\frac{1}{60}\Fourcov,
    \end{aligned}
\end{equation}
and the fourth-order quasi-entropy $\Xi_4$ is then written as
\begin{equation}
    \label{Simplified Xi_4}
    \begin{aligned}
        \Xi_4\left(Q,\,Q_4\right)=\,
         & -\log\det\Psi_2\left(Q+\bbi/3\right)-2\log\det \Psi_4(C_2)
        -\log\det(\Psi_4(C_1)-\Phi_2(Q)^t\Phi_2(Q))
        \\
        & 
        -2\log\det\begin{pmatrix}
                      \Psi_2\left(\bbi/3-Q/2\right)                    & \frac{1}{4}\Psi_3\left(\B\left(Q\right)\right) \\
                      \frac{1}{4}\Psi_3\left(\B\left(Q\right)\right)^t & \Psi_4(C_3)
                  \end{pmatrix},
    \end{aligned}
\end{equation}
where $\mathcal{B}{\left(Q\right)}_{ijk}=\epsilon_{ijs}Q_{ks}+\epsilon_{iks}Q_{js}$
and recall that $\epsilon$ is the Levi-Civita symbol.
The domain of $\Xi_4$ consists of $(Q,Q_4)$ making each the matrix after a log-determinant in \eqref{Simplified Xi_4} are positive definite, enforced by the barrier function property given by the log-determinant. 
The domain contains all the $(Q,Q_4)$ calculated from the average under some density
function, since the density function ensures that the covariance matrix is
positive definite.
Let us restate the properties of $\Xi_4$, which were discussed for general cases in \cite{XU2022133308}, Proposition 4.1 and Theorem 4.6. 
\begin{proposition}\label{pro:proposition of Xi_4}
    \begin{enumerate}
        \item The domain of $\Xi_4$ is convex, and $\Xi_4$ is strictly convex w.r.t. $Q,Q_4$.
        \item For any rotation $\mathfrak{t}\in SO(3)$, 
              $\Xi_4\left(Q,Q_4\right)=\Xi_4\left(\mathfrak{t}\circ Q,\mathfrak{t}\circ Q_4\right)$.
              In particular, \eqref{Simplified Xi_4} does not depend on the choice of $(\n_1,\n_2,\n_3)\in SO(3)$. 
    \end{enumerate}
\end{proposition}
A convenient choice of $(\n_1,\n_2,\n_3)$ would be the eigenframe of $Q$. The closure approximation by the quasi-entropy is given by substituting $\widetilde{\zeta}(Q,Q_4)$ with $\Xi_4(Q,Q_4)$ in the minimization problem
in \eqref{eq:the minimization problem of the Bingham closure for Q_4},
which reads
\begin{equation}
    \label{eq:the quasi-entropy closure}
    Q_4={\arg\min}\,\Xi_4\left(Q,Q_4\right),\quad Q \in \Qset\ \text{fixed}.
\end{equation}
The properties of the $Q_4$ obtained by the above quasi-entropy closure approximation result from the properties of the $\Xi_4$. 

\begin{theorem}\label{the:existence and uniqueness of quasi-entropy}
    For any $Q\in\Qset$, \eqref{eq:the quasi-entropy closure} yields a uniquely determined fourth-order tensor $Q_4$. Furthermore, $Q_4(Q)$ is continuous w.r.t. $Q$. 
\end{theorem}
\begin{proof}
  The uniqueness follows directly from the strict convexity of $\Xi_4$.

  For the existence, from \eqref{Simplified Xi_4} we deduce that
  $\Psi_4(C_i),\, (i=1,2,3)$ are positive definite.
  Meanwhile, 
  \begin{align*}
    -\Psi_4\left(\Fourcov \right)=\Psi_4(9C_1+24C_2+24C_3)=\diag\left(9,3,12,12,12\right).
  \end{align*}
  Those imply that the eigenvalues of $\Psi_4(C_i),\,(i=1,2,3)$ are bounded, so are the components of $C_i$.
  A direct consequence is that $\Xi_4$ is bounded from below.
  Moreover, $Q_4$ is bounded since we have the following identity, 
  \[Q_4 = \frac{18}{35}C_1+\frac{8}{35}C_2-\frac{32}{35}C_3.\]

  When $Q\in\Qset$, there exists a $Q_4$ such that $(Q,Q_4)$ lies within the
  domain, since we can choose a density function to obtain such a $Q_4$. So,
  there exists a sequence $\{Q_{4}^{k}\}$ such that
  $\lim_{k\to\infty}\Xi_4(Q,Q_4^{k})=\inf_{Q_4}\Xi_4(Q,Q_4)$. The boundedness of
  $Q_4^{k}$ implies a limit point $ \overline{Q_{4}} $. It is clear that $\Xi_4$
  is continuous w.r.t.  $Q_4$, so that $ \overline{Q_{4}} $ must lie
  within the domain of $\Xi_4$ (otherwise at least one of the determinants in
  $\Xi_4$ is zero, contradicting the equality above). As a result,
  $\Xi_4(Q,\overline{Q_4})$ attains the infimum, and $\overline{Q_4}$ gives the
  solution.

  For the continuity, let $ \{Q^k\} $ be a sequence that converges to $Q^0\in
  \Qset$, and define $Q_4^k = Q_4(Q^k)$. By the boundedness of $Q_4$, we can
  choose a subsequence $Q^{k_l}$ such that $\lim_{l\to
    \infty}Q_4(Q^{k_l})=\overline{Q_4}$. The facts that $\Xi_4$ is continuous
  w.r.t. $(Q,Q_4)$ and that $Q_4(Q^{k_l})$ is the minimizer with $Q=Q^{k_l}$
  fixed lead to
  \begin{align*}
    \Xi_4(Q,\overline{Q_4})=\lim_{l\to \infty}\Xi_4(Q^{k_l},Q_4(Q^{k_l}))
    \le \lim_{l\to \infty}\Xi_4(Q^{k_l},Q_4(Q))
    =\Xi_4(Q,Q_4(Q))<+\infty.
  \end{align*}
  It follows that $\overline{Q_4}=Q_4(Q)$, which concludes the proof.
\end{proof}

In addition, the rotational invariance of $\Xi_4$ leads to the fact that $Q_4$
rotates together with $Q$.
\begin{theorem}\label{th:q4 about n axis}
    For any $\mathfrak{t}\in SO(3)$, it holds $Q_4(\mathfrak{t}\circ Q)=\mathfrak{t}\circ Q_4(Q)$. When $ Q $ is written as \eqref{eq:the spectral decomposition of Q}, $ Q_4 $, determined by the quasi-entropy closure, takes the form
    \begin{equation}
        \label{eq:the a123 of Q4}
        Q_4(Q) = a_1 {\left( \n_1^4 \right)}_0 + a_2 {\left( \n_2^4 \right)}_0 + a_3 {\left( \n_1^2 \n_2^2 \right)}_0,
    \end{equation}
    where $ a_1, a_2, a_3 $ are scalar coefficients depending on $ s $ and $ b $. In the uniaxial case $b = 0$, we have $Q_4(Q) = a_1(s){\left( \n_1^4 \right)}_0$.
\end{theorem}

\begin{proof}

    The fourth-order symmetric traceless tensor $Q_4$ has nine degrees of freedom and can be represented, in the eigenframe of $Q$, as a linear
    combination of the basis:
    \begin{equation*}
        \begin{aligned}
            Q_4  =\, & a_1{\left(\n_1^4\right)}_0+a_2{\left(\n_2^4\right)}_0+a_3{\left(\n_1^2\n_2^2\right)}_0
            +a_4{\left(\n_1^3\n_2\right)}_0+a_5{\left(\n_1^3\n_3\right)}_0+a_6{\left(\n_1^2\n_2\n_3\right)}_0             \\
                     & +a_7{\left(\n_1\n_2^3\right)}_0+a_8{\left(\n_1\n_2^2\n_3\right)}_0+a_9{\left(\n_2^3\n_3\right)}_0.
        \end{aligned}
    \end{equation*}
    Here, the symmetric traceless tensors have the following form:
    \begin{equation}
        \begin{aligned}
             & {(\n_1^2\n_2^2)}_0 = \n_1^2\n_2^2 - \frac{1}{7}(\n_1^2+\n_2^2)\bbi+\frac{1}{35}\bbi^2,\
            {(\n_1^3\n_2)}_0 = \n_1^3\n_2 - \frac{3}{7}\n_1\n_2\bbi,
            \
            {(\n_1^2\n_2\n_3)}_0 = \n_1^2\n_2\n_3-\frac{1}{7}\n_2\n_3\bbi.
        \end{aligned}
    \end{equation}
    Define $\mathfrak{j}_1:(\n_1,\n_2,\n_3)\mapsto (\n_1,-\n_2,-\n_3)$ and similarly for $\mathfrak{j}_2$, $\mathfrak{j}_3$, which satisfy $\mathfrak{j}_l\circ Q=Q$.
    By convexity,
    \begin{equation*}
        \begin{aligned}
                 & \Xi_4\left(Q,\frac{1}{4}\left(Q_4+\mathfrak{j}_1\circ Q_4+\mathfrak{j}_2\circ Q_4+\mathfrak{j}_3\circ Q_4\right)\right)                                                             \\
            \leq & \frac{1}{4}\left(\Xi_4\left(Q,Q_4\right)+\Xi_4\left(Q,\mathfrak{j}_1\circ Q_4\right)+\Xi_4\left(Q,\mathfrak{j}_2\circ Q_4\right)+\Xi_4\left(Q,\mathfrak{j}_3\circ Q_4\right)\right)
            =\Xi_4\left(Q,Q_4\right).
        \end{aligned}
    \end{equation*}
    The uniqueness of $Q_4$ implies $Q_4=\frac{1}{4}\left(Q_4+\mathfrak{j}_1\circ Q_4+\mathfrak{j}_2\circ Q_4+\mathfrak{j}_3\circ Q_4\right)=a_1{\left(\n_1^4\right)}_0+a_2{\left(\n_2^4\right)}_0+a_3{\left(\n_1^2\n_2^2\right)}_0$.

    In the uniaxial case, we consider $\mathfrak{b}_1:\big(\n_1,\n_2,\n_3)\mapsto
    (\n_1,(\n_2+\n_3)/\sqrt{2},(\n_3-\n_2)/\sqrt{2}\big)$.
    Using an argument similar to the above, we deduce that
    $Q_4=\frac{1}{4}\left(Q_4+\mathfrak{b}_1 \circ Q_4+\mathfrak{b}_1^2 \circ Q_4+\mathfrak{b}_1^3\circ Q_4\right)$,
    from which, by substituting \eqref{eq:the a123 of Q4} and using the equalities $\bbi=\n_1^2+ \n_2^2 + \n_3^2$ and
    ${\left(\n_1^2\bbi\right)}_0=\mathbf{0}$~\cite{Xu_2020}, we arrive at
    $Q_4=a_1(\n_1^4)_0$.
\end{proof}

The fourth-order tensor $\fM $ is defined by
\eqref{eq:the dissipative operator M}
with $Q_4$ obtained from \eqref{eq:the quasi-entropy closure},
denoted as $\fM ^{\mathrm{qent}}$.
Its existence, uniqueness and continuity follow from those of $Q_4$. 
Moreover, $\fM ^{\mathrm{qent}}$ has the same form as $\fM^{\mathrm{Bin}}$ when expressed by the eigenframe of $Q$. 

\begin{corollary}\label{pro:properties of M_QC}
    $\fM ^{\mathrm{qent}}$ rotates along with $Q$ as $\fM ^{\mathrm{qent}}\!\left(\mathfrak{t}\circ Q\right) = \mathfrak{t}\circ \fM ^{\mathrm{qent}}(Q)$ for any \(\mathfrak{t}\in SO(3)\).
    Moreover, when $Q$ is written as \eqref{eq:the spectral decomposition of Q},
    $\fM ^{\mathrm{qent}}$ has the form \eqref{eq:eigen-M}.
    In the uniaxial case (i.e., $b=0$), the fourth-order tensor $\fM
        ^{\mathrm{qent}}$ takes the form \eqref{eq:uniaxial-MBin-n1}.
\end{corollary}

\begin{proof}
    $\fM ^{\mathrm{qent}}\!\left(\mathfrak{t}\circ Q\right) = \mathfrak{t}\circ \fM ^{\mathrm{qent}}(Q)$ follows immediately from the similar argument for $Q_4$.
    Substituting~\eqref{eq:the a123 of Q4} into \eqref{eq:the dissipative operator M}--\eqref{eq:symtrls-M},
    we arrive at
    \begin{equation}\label{M-sbasis}
        \fM ^{\mathrm{qent}} =-\frac{2}{15}\Fourcov  - \frac{4}{7}s\mathcal{A}({(\n_1^2)}_0) - \frac{4}{7}b\mathcal{A}(\n_2^2-\n_3^2) -4(a_1{(\n_1^4)}_0+a_2{(\n_2^4)}_0+a_3{(\n_1^2\n_2^2)}_0).
    \end{equation}
    When expressing the six tensors on the right-hand side under the basis $\s_i\otimes\s_j$, they all take the form \eqref{eq:eigen-M}. We give the explicit expressions shortly afterwards.
    In the uniaxial case, Theorem~\ref{th:q4 about n axis} gives 
    $b=a_2=a_3=0$. Direct calculations using the definition of $(\n_1^2)_0$ and
    $(\n_1^4)_0$ yield the form \eqref{eq:uniaxial-MBin-n1}.
\end{proof}
In the end, the domain of the quasi-entropy $\Xi_4$ guarantees the positive definiteness of $\fM ^{\mathrm{qent}}$.
\begin{theorem}\label{the:positive of M in quasi-entropy}
    $\fM ^{\mathrm{qent}}$ is positive definite.
\end{theorem}
\begin{proof}
    Note that $\fM ^{\mathrm{qent}}=8C_3$ where $C_3$ is defined in \eqref{eq:the three matrices of Psi_4}.
    In $\Xi_4$, the matrix $\Psi_4(C_3)$ appears as a diagonal block in a log-determinant, so that $\fM ^{\mathrm{qent}}$ is positive definite.
\end{proof}

When $(\n_1,\n_2,\n_3)\in SO(3)$ takes an eigenframe of $Q$, let us write down the quasi-entropy $\Xi_4$. 
We calculate $\Psi_4(E)$, $\Psi_4\big(\A (\n_1^2)_0\big)$, $\Psi_4\big(\A (\n_2^2-\n_3^2)\big)$, $\Psi_4\big((\n_1^4)_0\big)$, $\Psi_4\big((\n_2^4)_0\big)$, $\Psi_4\big((\n_1^2\n_2^2)_0\big)$, yielding  
\begin{equation}
    \label{eq:the coordinate matrices of 6tensors}
    \begin{aligned}
         & X_1=\diag\left(\begin{pmatrix}
                                  -9 & 0  \\
                                  0  & -3
                              \end{pmatrix},-12,-12,-12\right),\,
        X_2=\diag\left(\begin{pmatrix}
                               -\frac{3}{2} & 0           \\
                               0            & \frac{1}{2}
                           \end{pmatrix},-1,-1,2\right)                                          \\
         & X_3=\diag\left(\begin{pmatrix}
                                  0           & \frac{3}{2} \\
                                  \frac{3}{2} & 0
                              \end{pmatrix},-3,3,0\right),\,
        X_4=\diag\left(\begin{pmatrix}
                           \frac{18}{35} & 0            \\
                           0             & \frac{1}{35}
                       \end{pmatrix},-\frac{16}{35},-\frac{16}{35},\frac{4}{35}\right)       \\
         & X_5=\diag\left(\begin{pmatrix}
                              \frac{27}{140} & -\frac{3}{28}  \\
                              -\frac{3}{28}  & \frac{19}{140}
                          \end{pmatrix},-\frac{16}{35},\frac{4}{35},-\frac{16}{35}\right),\,
        X_6=\diag\left(\begin{pmatrix}
                           -\frac{9}{35} & \frac{3}{28}  \\
                           \frac{3}{28}  & -\frac{1}{70}
                       \end{pmatrix},\frac{18}{35},-\frac{2}{35},-\frac{2}{35}\right).
    \end{aligned}
\end{equation}
Define
\begin{align*}
     & \omega={\left(s,b,0,0,0\right)}^T,\quad W_1\equaldef\Psi_2\left(Q+\frac{1}{3}\bbi\right)=\diag\left(\frac{1}{3}\left(2s+1\right),\frac{1}{3}\left(1-s\right)+b,\frac{1}{3}\left(1-s\right)-b\right), \\
     & W_2=W_3\equaldef\Psi_2\left(\frac{1}{3}\bbi-\frac{1}{2}Q\right)=\diag\left(\frac{1}{3}\left(1-s\right),\frac{1}{6}\left(2+s\right)-\frac{1}{2}b,\frac{1}{6}\left(2+s\right)+\frac{1}{2}b\right),     \\
     & T_2=\begin{pmatrix}
               0 & 0 & 0                            & 0                           & -b \\
               0 & 0 & 0                            & \frac{1}{2}\left(s+b\right) & 0  \\
               0 & 0 & \frac{1}{2}\left(-s+b\right) & 0                           & 0
           \end{pmatrix}.
\end{align*}
Then, the quasi-entropy $\Xi_4$ is written as 
\begin{equation}
    \label{Xi4 ss}
    \begin{aligned}
        \Xi_4\left(s,b,a_1,a_2,a_3\right) =
         & -\log\det\left(-\frac{1}{45}X_1-\frac{4}{21}\left(sX_2+bX_3\right)+a_1X_4+a_2X_5+a_3X_6 -\omega\omega^T\right)          \\
         & -2\log\det \left(-\frac{1}{15}X_1+\frac{4}{7}\left(sX_2+bX_3\right)+\frac{1}{2}\left(a_1X_4+a_2X_5+a_3X_6\right)\right) \\
         & -2\log\det\begin{pmatrix}
                         W_2
                               & T_2                                                                                              \\
                         T_2^T & -\frac{1}{60}X_1-\frac{1}{14}\left(sX_2+bX_3\right)-\frac{1}{2}\left(a_1X_4+a_2X_5+a_3X_6\right)
                     \end{pmatrix}      \\
         & -\log\det W_1 + c_0,
    \end{aligned}
\end{equation}
for some constant $c_0$. 
In the first two lines, each matrix is $5\times 5$ consisting of diagonal blocks of $2\times 2$ and three $1\times 1$.
In the third line, the matrix is $8\times 8$ and can be rearranged so that it consists of four $2\times 2$ diagonal blocks.

By Theorem~\ref{th:q4 about n axis} the closure approximation can be done according to the following three steps.
\begin{enumerate}
\item Given $Q$, find the $s$, $b$, and the eigenframe $(\n_1,\n_2,\n_3)\in SO(3)$. 
\item Solve $(a_1,a_2,a_3)$ by minimizing $\Xi_4(s,b,a_1,a_2,a_3)$ given in \eqref{Xi4 ss} with $(s,b)$ fixed. 
\item Assemble $Q_4$ and $\fM ^{\mathrm{qent}}$ using \eqref{eq:the a123 of Q4} and \eqref{eq:the dissipative operator M}. 
\end{enumerate}
In the second step, we only need to minimize an elementary function, which is easy to implement. 
In particular, the block diagonal structure of \eqref{Xi4 ss} makes it convenient to evaluate the function value and its derivatives and to determine whether certain $(s,b,a_1,a_2,a_3)$ lies within the domain of $\Xi_4$.

\subsection{Gradient flow with the quasi-entropy}
In \eqref{eq:gradient flow of Bingham closure}, we use $\Xi_2(Q)$ to substitute $\zeta(Q)$ in the free energy and let the tensor $\fM$ be given by $\fM ^{\mathrm{qent}}$.
In this way, we obtain the following gradient flow, 
\begin{equation}
    \label{eq:DTM with quasi-entropy}
    \begin{aligned}
        \frac{\partial Q}{\partial t} & = -\fM ^{\mathrm{qent}} \mu_Q,\quad \mu_Q=\frac{\delta F}{\delta Q}.
    \end{aligned}
\end{equation}
The free energy is given by
\begin{equation}
    \label{eq:free energy of the DTM with quasi-entropy}
    F[Q]=\int \Xi_2\left(Q\right)\md\bx+F_{\mathrm{r}}[Q],
\end{equation}
where the interaction energy $F_r$ contains the bulk and the elastic energy (involving spatial derivatives), 
\begin{equation}\label{eq:interaction energy}
    F_{\mathrm{r}}=\frac{1}{2}\int -c_{02}\left|Q\right|^2+c_{21}\left|\nabla Q\right|^2+c_{22}\partial_{i}Q_{ik}\partial_{j}Q_{jk}+c_{24}Q_{ij}\partial_{i}Q_{kl}\partial_{j}Q_{kl}\,\md\bx,
\end{equation}
with the short notations $|Q|^2=Q\cdot Q$, $|\nabla Q|^2=\partial_iQ_{jk}\partial_iQ_{jk}$.
We assume that, after integration by parts, the boundary terms either vanish or contribute only a constant.
This assumption can be satisfied under the periodic or Dirichlet boundary conditions, for which we omit the detailed discussions.
It is worth noting that the second-order quasi-entropy $\Xi_2$ gives a singular term in the free energy that is able to constrain the eigenvalues of $Q$.
This enables us to include in the free energy a cubic elastic term while the free energy is bounded from below.
Without this constraint, it has been shown that this cubic term is unbounded from below~\cite{Ball_Majumdar_2010}.
\begin{theorem}\label{thm:lower bound of the energy}
    Suppose that $c_{21}>0, c_{21}+\min\left \{ c_{22},0\right \} \geq\max\left \{ \frac{1}{3}c_{24},-\frac{2}{3}c_{24}\right \} $ and $Q\in\Qset$, the free energy~\eqref{eq:free energy of the DTM with quasi-entropy} is bounded from below.
\end{theorem}
\begin{proof}
    Using integration by parts, we rewrite $|\nabla Q|^2$ as the sum of two nonnegative terms,
    \begin{equation}\label{eq:grad_div}
        \int \left|\nabla Q\right|^2\,\md\bx=\int \frac{1}{2}\left(\partial_i Q_{jk}-\partial_j Q_{ik}\right)\left(\partial_i Q_{jk}-\partial_j Q_{ik}\right)+\partial_{i}Q_{ik}\partial_j Q_{jk}\,\md\bx
        \ge
        \int \partial_{i}Q_{ik}\partial_j Q_{jk}\,\md\bx.
    \end{equation}
    For the cubic term, we have
    \begin{equation}\label{eq:boundness of cubic term}
        \begin{aligned}
             & \frac{1}{3}|\nabla Q|^2+Q_{ij}\partial_iQ_{kl}\partial_jQ_{kl}=(Q_{ij}+\frac{1}{3}\delta_{ij})\partial_iQ_{kl}\partial_jQ_{kl}, \\
             & \frac{2}{3}|\nabla Q|^2-Q_{ij}\partial_iQ_{kl}\partial_jQ_{kl}=(\frac{2}{3}\delta_{ij}-Q_{ij})\partial_iQ_{kl}\partial_jQ_{kl},
        \end{aligned}
    \end{equation}
    which are nonnegative when $Q\in\Qset$. The proof is concluded following the above equations.
\end{proof}

The energy dissipation law follows from the fact that $\fM ^{\mathrm{qent}}$ is
positive definite.
\begin{theorem}
    For the gradient flow~\eqref{eq:DTM with quasi-entropy}, it holds the following energy dissipation law, 
    \begin{equation}
        \begin{aligned}
            \frac{\md F}{\md t}=-\int \mu_Q\cdot\fM ^{\mathrm{qent}}\mu_Q\,\md\bx\leq 0.
        \end{aligned}
    \end{equation}
\end{theorem}

\section{Numerical simulations}\label{sec:numerical-method}
As we have discussed above, in tensor dynamics (no matter with the Bingham closure or the quasi-entropy) the eigenvalue constraints and the dissipation structure play key roles. 
Maintaining these properties thus becomes the major target computationally. 
When adopting the quasi-entropy, we are able to construct numerical schemes in a natural way to take care of these properties. 
To illustrate this convenience, we construct a temporal first-order scheme for \eqref{eq:DTM with quasi-entropy} satisfying these properties. 
In computations, we shall pay attention to these properties and the effects of the fourth-order tensor $\fM$ and the $c_{24}$ term.

\subsection{Numerical methods}

The most significant point in a numerical scheme for \eqref{eq:DTM with quasi-entropy} is to keep $Q\in \Qset$, because it is necessary for the free energy and the closure approximation to be defined. 
In the gradient flow, it is achieved by the barrier given by the quasi-entropy $\Xi_2$.
Therefore, we choose to discretize $\partial \Xi_2/\partial Q$ implicitly. 
Because the fourth-order tensor $\fM ^{\mathrm{qent}}$ requires solving a minimization problem,
it is preferable to deal with it explicitly in the time discretization. 
In this way, the closure approximation is decoupled from the solving procedure of the scheme, so that its easy implementation described in the previous section is inherited in the scheme. 

Moreover, since the fourth-order tensor $\fM ^{\mathrm{qent}}$ from the explicit discretization is positive definite, it suffices to find an appropriate way to discretize $\mu_Q=\delta F/\delta Q$ to maintain the dissipation structure. 
Here, we choose the convex splitting method since it ensures the existence and uniqueness of the scheme. 
Consider the splitting of the free energy $F[Q] = F_{+}[Q] - F_{-}[Q]$ in the following form, 
\begin{equation}\label{eq:the convex splitting}
    \begin{aligned}
        F_{+}[Q]
        = & \int 
        \Xi_2(Q)
        +\frac{1}{2}\Bigl(
        \gamma_1\left|\nabla Q\right|^2
        +c_{22}\partial_i Q_{ik}\partial_j Q_{jk}
        \Bigr)                       \\
          & \quad
        +\frac{c_{24}}{2}
        Q_{ij}\partial_i Q_{kl}\partial_j Q_{kl}
        +\frac{\gamma_2}{2}(\left|Q\right|^2
        +\frac{1}{2}\left|\nabla Q\right|^4)
        +\frac{\gamma_3}{2}\left|\nabla Q\right|^2
        \,\md\bx                     \\
        F_{-}[Q]
        = & \frac{1}{2}\int 
        (c_{02}+\gamma_2)\left|Q\right|^2
        +\frac{\gamma_2}{2}\left|\nabla Q\right|^4
        +(\gamma_1+\gamma_3-c_{21})\left|\nabla Q\right|^2
        \,\md\bx,
    \end{aligned}
\end{equation}
where $\gamma_i,\,(i=1,2,3)$ are to be determined to make $F_{\pm}[Q]$ convex w.r.t. $Q$.
Denote $\mu_{\pm}=\delta F_{\pm}/\delta Q$. 
The scheme is given by 
\begin{equation}\label{eq:first order scheme derivation}
    \frac{Q^{n+1} - Q^{n}}{\delta t}
    = - \fM ^{\mathrm{qent},n}\mu^{n+1},\quad \mu^{n+1}=\mu_{+}(Q^{n+1}) - \mu_{-}(Q^{n}).
\end{equation}

It is clear that each term in $F_{-}$ is convex given that the coefficient of each term is positive. 
For $F_{+}$, the convexity is established under certain conditions on the coefficients. 
\begin{theorem}\label{the:convex of F}
    If $\gamma_1+\min\{c_{22},0\}\ge 0$, $\gamma_2\ge |c_{24}|$, $\gamma_3\ge \max\{c_{24}/3,-2c_{24}/3\}$, then $F_+$ is a convex functional on $Q\in\Qset$.
\end{theorem}
\begin{proof}
  The quasi-entropy \(\Xi_2\) is convex by Proposition~\ref{propsition:quasi-entropy xi2}. 
  The convexity of the parentheses including the $c_{22}$ follows from \eqref{eq:grad_div}.
  It remains to examine the $c_{24}$ term. 
  Denote $h(Q)=Q_{ij}\partial_iQ_{kl}\partial_jQ_{kl}$.
  For \(Q,Q\pm\tilde{Q}\in\Qset\), we calculate 
  \begin{equation}
    h(Q+\tilde{Q})+h(Q-\tilde{Q})-2h(Q)=Q_{ij}\partial_i\tilde{Q}_{kl}\partial_j\tilde{Q}_{kl}+2\tilde{Q}_{ij}\partial_i\tilde{Q}_{kl}\partial_jQ_{kl}. 
    \end{equation}
  The first term is controlled by $\frac{\gamma_2}{2}|\nabla Q|^2$
  according to inequalities similar to \eqref{eq:boundness of cubic term}.
  The second term can be controlled by the inequality
    \begin{equation}
        2|\tilde{Q}_{ij}\partial_{i}\tilde{Q}_{kl}\partial_{j}Q_{kl}| \leq \tilde{Q}_{ij}\tilde{Q}_{ij} + \partial_{i}\tilde{Q}_{kl}\partial_{j}{Q}_{kl}\partial_{i}\tilde{Q}_{k^{'}l^{'}}\partial_{j}{Q}_{k^{'}l^{'}}\leq |\tilde{Q}|^2 + |\nabla Q|^2|\nabla\tilde{Q}|^2,
    \end{equation}
    where the second inequality is deduced from 
    \begin{equation*}
        (\partial_{i} \tilde{Q}_{kl}\partial_{j}Q_{k^{'}l^{'}} - \partial_{i}\tilde{Q}_{k^{'}l^{'}}\partial_{j}Q_{kl} )(\partial_{i} \tilde{Q}_{kl}\partial_{j}Q_{k^{'}l^{'}} - \partial_{i}\tilde{Q}_{k^{'}l^{'}}\partial_{j}Q_{kl} ) \geq 0.
    \end{equation*}
    The convexity then follows from
    \begin{align*}
      &2|\tilde{Q}|^2=|Q+\tilde{Q}|^2+|Q-\tilde{Q}|^2-2|Q|^2, \\
      &4|\nabla Q|^2|\nabla\tilde{Q}|^2
      \le |\nabla(Q+\tilde{Q})|^4+|\nabla(Q-\tilde{Q})|^4-2|\nabla Q|^4. 
    \end{align*}
\end{proof}

Under appropriate assumptions on $\fM ^{\mathrm{qent},n}$, the resulting scheme admits a unique solution at each time step. 
\begin{theorem}\label{thm:exist-unique}
    Let $Q^{n}\in \Qset$.
    Assume that
    $\operatorname{ess}\inf \lambda\big(\Psi_4(\fM ^{\mathrm{qent},n})\big)>0$.
    Then the scheme~\eqref{eq:first order scheme derivation} admits a unique solution
    $Q^{n+1}\in \Qset$.
    If $Q^{n}$ is a stationary point of the free energy, i.e.,
    $\delta F/\delta Q(Q^{n})=0$, then the scheme~\eqref{eq:first order scheme derivation}
    leaves it unchanged: $Q^{n+1}=Q^{n}$.
\end{theorem}
\begin{proof}
    For existence and uniqueness, consider the functional
    \begin{equation}
        H\!\left[Q^{n+1}\right]
        = \int 
        \left[
            \frac{1}{2\delta t}{\left(\fM^{\mathrm{qent},n}\right)}^{-1}(Q^{n+1} - Q^{n})\cdot(Q^{n+1} - Q^{n})
            - Q^{n+1}\cdot \mu_{-}\!\left(Q^{n}\right)
            \right]\,\md\bx
        + F_{+}\!\left(Q^{n+1}\right),
    \end{equation}
    where the inverse is defined by \((\fM^{\mathrm{qent,n}})^{-1}_{ijkl}\fM^{\mathrm{qent},n}_{kli^{'}j^{'}} = \delta_{ii^{'}}\delta_{jj^{'}}\).
    The proof follows essentially the same procedure as in~\cite{Wang_Xu_2023}.
    The only additional ingredient here is the presence of the
    ${\left(\fM ^{\mathrm{qent},n}\right)}^{-1}$.
    The assumption of a positive lower bound for $\lambda\big(\Psi_4(\fM ^{\mathrm{qent},n})\big)$ ensures that the functional $H[Q]$ is bounded from below, and hence the derivations in~\cite{Wang_Xu_2023} still hold. 

    If $\delta F/\delta Q(Q^{n})=0$, then $\mu_{+}(Q^{n})-\mu_{-}(Q^{n})=0$.
    Substituting $Q^{n+1}=Q^{n}$ into~\eqref{eq:first order scheme derivation}
    yields an identity.
    Together with the uniqueness, $Q^{n+1}=Q^{n}$ is the solution. 
\end{proof}

\begin{theorem}\label{thm:discrete-dissipation}
    The scheme~\eqref{eq:first order scheme derivation} satisfies the dissipation law
    \[
        F[Q^{n+1}] \le F[Q^{n}] -\int \mu^{n+1}\cdot\fM ^{\mathrm{qent,n}}\mu^{n+1}\md\bx.
    \]
\end{theorem}

\begin{proof}
  By the convexity of $F_+,\,F_-$, we have
  \begin{align*}
    \int \mu_{+}(Q^{n+1})\cdot(Q^{n+1} - Q^{n})\,\md \bx \geq F_{+}[Q^{n+1}] - F_{+}[Q^n], \\
    \int -\mu_{-}(Q^{n})\cdot(Q^{n+1} - Q^{n})\,\md \bx \geq F_{-}[Q^{n}] - F_{-}[Q^{n+1}].
  \end{align*}
  Multiplying \eqref{eq:first order scheme derivation} by $\mu^{n+1}$ and integrating, we obtain the dissipation law from the positive definiteness of $\fM ^{\mathrm{qent},n}$.
\end{proof}

To compute the quasi-entropy closure approximation, it is necessary to provide an initial guess of $(a_1,a_2,a_3)$ for given $(s,b)$ such that they lie within the domain of $\Xi_4$, which is provided below. 
\begin{theorem}\label{the:the initial value of Q4}
    Let $\theta_i\in (0,1),\,(i=1,2,3)$ be the eigenvalues of $Q+\bbi/3$ (determined by $s$, $b$). For
    \begin{equation}
        \label{eq:the initial value of quasi-entropy}
        a_1=\theta_1+\theta_3-\theta, \quad
        a_2=\theta_2+\theta_3-\theta, \quad
        a_3=2\theta_3-\theta,\qquad
        \theta=\min\left \{ \theta_1,\theta_2,\theta_3\right \}, 
    \end{equation}
    $(s,b,a_1,a_2,a_3)$ lies within the domain of $\Xi_4$, i.e. makes all the matrices in \eqref{Xi4 ss} positive definite.
\end{theorem}
\begin{proof}
  The $a_1,a_2,a_3$ in \eqref{eq:the initial value of quasi-entropy} are obtained by constructing a discrete density function \(\tilde{f}(\m)\) on $\mathbb{S}^2$ such that $\langle(\m^2)_0\rangle_{\tilde{f}}=Q$. 
  Let \(
  \tilde{f}(\m)=\sum_{i=1}^{9}\alpha_i\delta(\m-\mathbf{a}_i),\,\alpha_i>0,\, \sum_{i=1}^{9}\alpha_i=1\),
  where \( \mathbf{a}_i=\n_i,\, i=1,2,3\), \(
  \mathbf{a}_{4,7}=\frac{\sqrt{2}}{2}(\n_1\pm\n_2),
  \mathbf{a}_{5,8}=\frac{\sqrt{2}}{2}(\n_1\pm\n_3),
  \mathbf{a}_{6,9}=\frac{\sqrt{2}}{2}(\n_2\pm\n_3) \).
  Set \(\alpha_4 = \cdots = \alpha_9=\alpha\).
  The equalities $\alpha_i+2\alpha=\theta_i$ need to hold.
  To this end, we set \(\theta=\min\{\theta_1,\theta_2,\theta_3\},\,\alpha=\theta/5\), so that \(\alpha_i=\theta_i-\frac{2}{5}\theta, i=1,2,3\).
  It is easy to verify that $\alpha_i>0$ for $1\le i\le 9$. 
  Calculating $\langle(\m^4)_0\rangle_{\tilde{f}}$ yields \eqref{eq:the initial value of quasi-entropy}. 

  Next, we show that \eqref{eq:the initial value of quasi-entropy} makes all the matrices in \eqref{Xi4 ss} positive definite, for which it suffices to verify for the $2\times 2$ and $1\times 1$ blocks. Notice that $\theta_1+\theta_2+\theta_3=1$. 
  Let us handle the $2\times 2$ blocks first. 
  The matrix in the first line of \eqref{Xi4 ss} has one \(2\times 2\) block, 
  \[
  \begin{pmatrix}
    \frac{9}{4}\theta_1-\frac{9}{4}\theta_1^2-\frac{9}{20}\theta & \frac{3}{4}\theta_1(\theta_{3}-\theta_{2})                                      \\
    \frac{3}{4}\theta_1(\theta_{3}-\theta_{2})                   & \frac{1}{4}(1-\theta_1) - \frac{1}{4}(\theta_2 - \theta_3)^2-\frac{3}{20}\theta
  \end{pmatrix}.
  \]
  Since $\theta_1\in [\theta,1-2\theta]$ and $\theta\le 1/3$, it holds for the upper-left element and the determinant that 
  \begin{equation}
    \begin{aligned}
      & \frac{9}{4}\theta_1 - \frac{9}{4}\theta_1^2 - \frac{9}{20}\theta\ge \frac{9}{4}\theta(1-\theta)-\frac{9}{20}\theta=\frac{9}{4}\theta(\frac{4}{5}-\theta)>0,\,                                                                                                                                                          \\
      & \frac{9}{16}\left(\frac{3}{25}\theta^2 + \frac{4}{5}\bigl(\theta_1\theta_3(\theta_2 - \theta)+\theta_2\theta_3(\theta_1-\theta)+3\theta_1\theta_2(\theta_3 - \frac{2}{3}\theta)+\theta(1-\theta_1)(1-\theta_2)\bigr)\right)>0.
    \end{aligned}
  \end{equation}
  The matrix in the second line of \eqref{Xi4 ss} has one \(2\times 2\) block, 
  \begin{equation}
    \begin{pmatrix}
      \frac{9}{8}(1-\theta_1-\frac{1}{5}\theta) & \frac{3}{8}(\theta_2-\theta_3)                      \\
      \frac{3}{8}(\theta_2-\theta_3)            & \frac{3}{8}(\theta_1+\frac{1}{3}-\frac{1}{5}\theta)
    \end{pmatrix}.
  \end{equation}
  The upper-left element and determinant satisfy 
  \begin{equation*}
    \frac{9}{8}(\theta_2 + \theta_3 - \frac{1}{5}\theta) > 0, \,\frac{3}{25}\theta^2 + \frac{16}{5}\theta + 4\theta_{1}(\theta_2 - \theta) + 4\theta_2(\theta_3 - \theta)+4\theta_3(\theta_1 - \theta) > 0.
  \end{equation*}
  The $8\times 8$ matrix in \eqref{Xi4 ss} can be rearranged to consist of four \(2\times 2\) blocks.
  Under \eqref{eq:the initial value of quasi-entropy} one $2\times 2$ block becomes a diagonal matrix, while the other three read
  \begin{equation*}
    \begin{pmatrix}
      \frac{1}{2}(1-\theta_1)          & \frac{1}{2}(\theta_3 - \theta_2)            \\
      \frac{1}{2}(\theta_3 - \theta_2) & \frac{1}{2}(1-\theta_1 - \frac{2}{5}\theta)
    \end{pmatrix},\begin{pmatrix}
      \frac{1}{2}(1-\theta_2)          & \frac{1}{2}(\theta_1 - \theta_3)            \\
      \frac{1}{2}(\theta_1 - \theta_3) & \frac{1}{2}(1-\theta_2 - \frac{2}{5}\theta)
    \end{pmatrix},\begin{pmatrix}
      \frac{1}{2}(1-\theta_3)          & \frac{1}{2}(\theta_2 - \theta_1)            \\
      \frac{1}{2}(\theta_2 - \theta_1) & \frac{1}{2}(1-\theta_3 - \frac{2}{5}\theta)
    \end{pmatrix}.
  \end{equation*}
  For the first matrix (similarly for the other two), we verify that its upper-left element and determinant that
  \begin{equation*}
    \begin{aligned}
      \frac{1}{2}(1-\theta_1) > 0, \, \frac{1}{10}(\theta_2 - \theta)(\theta_3 - \theta)+\frac{9}{10}\theta_2\theta_3 - \frac{1}{10}\theta^2>0.
    \end{aligned}
  \end{equation*}

  All the remaining \(1\times 1\) blocks are
  \begin{equation*}
    \frac{9}{40}\theta,\ \frac{3}{40}\theta,\ \frac{1}{5}\theta + 2\theta_3,\ \frac{1}{5}\theta + 2\theta_2,\ \frac{1}{5}\theta + 2\theta_1 ,\ \frac{2}{5}\theta ,\ \theta_1,\ \theta_2,\ \theta_3,
  \end{equation*}
  which are obviously positive. 
\end{proof}

\subsection{Numerical results}
\begin{figure}[htbp]
    \centering
    \includegraphics[width=0.4\textwidth]{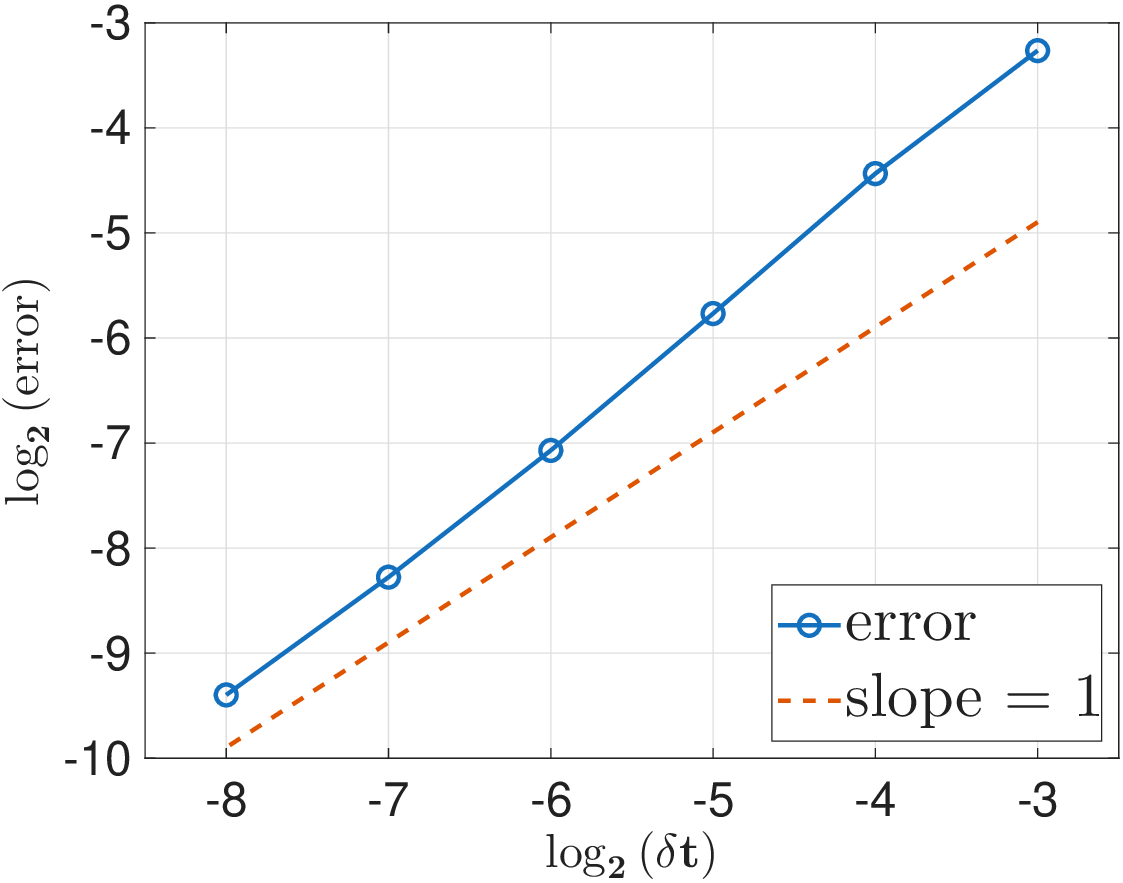}
    \caption{The error for scheme \eqref{eq:first order scheme derivation} with different $\delta t$.}\label{fig:First order time rate}
\end{figure}

We examine the case where the system is homogeneous in the $z$-direction, 
so that $Q(\bx)$ depends only on $(x,y)$.
The computational domain $\Omega = [0,2\pi]\times[0,2\pi]$ is equipped with periodic boundary conditions. 
For $\gamma_1$, $\gamma_2$, $\gamma_3$ in the scheme \eqref{eq:the convex splitting}--\eqref{eq:first order scheme derivation}, they are chosen so that inequalities hold in Theorem \ref{the:convex of F}. 
The discretization in space is carried out using the Fourier expansion with $N=32$ Fourier modes.
The nonlinear equations resulting from the scheme are solved by Newton's iteration with the tolerance $10^{-6}$.

For the influences of the fourth-order tensor $\fM$, 
we will compare the results with the $L^2$ gradient flow without incorporating the fourth-order tensor $\fM$, written as 
\begin{equation}\label{eq:L2gf}
  \frac{\partial Q}{\partial t}=-\mathcal{P}\mu_Q. 
\end{equation}
Here, $\mu_Q$ is still the variational derivative of $F$ given in \eqref{eq:free energy of the DTM with quasi-entropy}--\eqref{eq:interaction energy},
and $\mathcal{P}$ is the projection onto the symmetric traceless tensors, i.e. $(\mathcal{P}U)_{ij}=(U_{ij}+U_{ji})/2-U_{kk}\delta_{ij}/3$.
The first-order scheme proposed in \cite{Wang_Xu_2023} is adopted to solve \eqref{eq:L2gf}. 
Meanwhile, we are also interested in the influence of the cubic elastic $c_{24}$ term.
Accordingly, we define labels for different cases:
$\fM^{\mathrm{qent}}$ or $L^2$ indicates whether the equation \eqref{eq:DTM with quasi-entropy} or \eqref{eq:L2gf} we are solving;
Cub or Quad indicates $c_{24}\ne 0$ or $c_{24}=0$ (if $c_{24}$ is specified as a nonzero value, using the label Quad means that $c_{24}$ is set to zero). 
For example, $\fM ^{\mathrm{qent}}$--Cub means the results for \eqref{eq:DTM with quasi-entropy} with some nonzero $c_{24}$. 

\textbf{Accuracy.} We first carry out an accuracy test as a validation. Choose the coefficients
$c_{02}=15$, $c_{21}=0.16$, $c_{22} = 0.02$, $c_{24} = 0.015$,
and the initial condition,
\begin{equation}
    Q(x,y)=\big(s_0+0.05\sin(x)\sin(y)\big)\left(\e_1^2-\bbi/3\right),
\end{equation}
where $s_0$ is the minimizer of the bulk energy \eqref{eq:bulk energy}. 
The solution at $t=1$ is computed using various time steps. 
The reference solution is obtained by $\delta t=2^{-10}$. 
The error plotted in Figure~\ref{fig:First order time rate}
indicates first-order accuracy. 

\begin{figure}[htbp]
    \centering
    \includegraphics[width=1\textwidth]{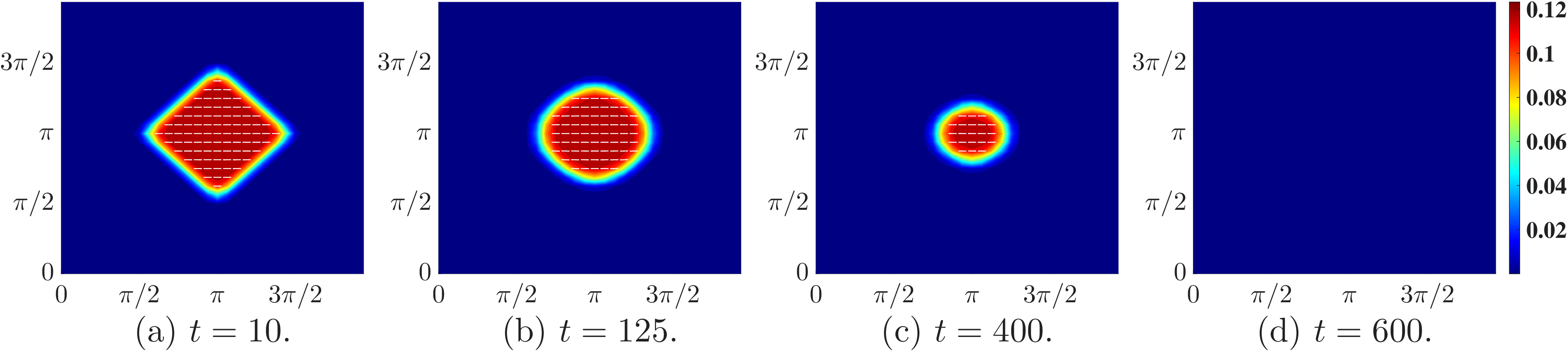}
    \caption{Evolution of the interface associated with the maximum eigenvalue of $Q$ in $\fM ^{\mathrm{qent}}$--Cub.}\label{fig:The evolution of the interface in Case 1}
\end{figure}
\begin{figure}[htbp]
    \centering
    \subfigure[The area enclosed by the $0.1$-contour of maximum eigenvalue of $Q$.]{\includegraphics[width=0.31\textwidth]{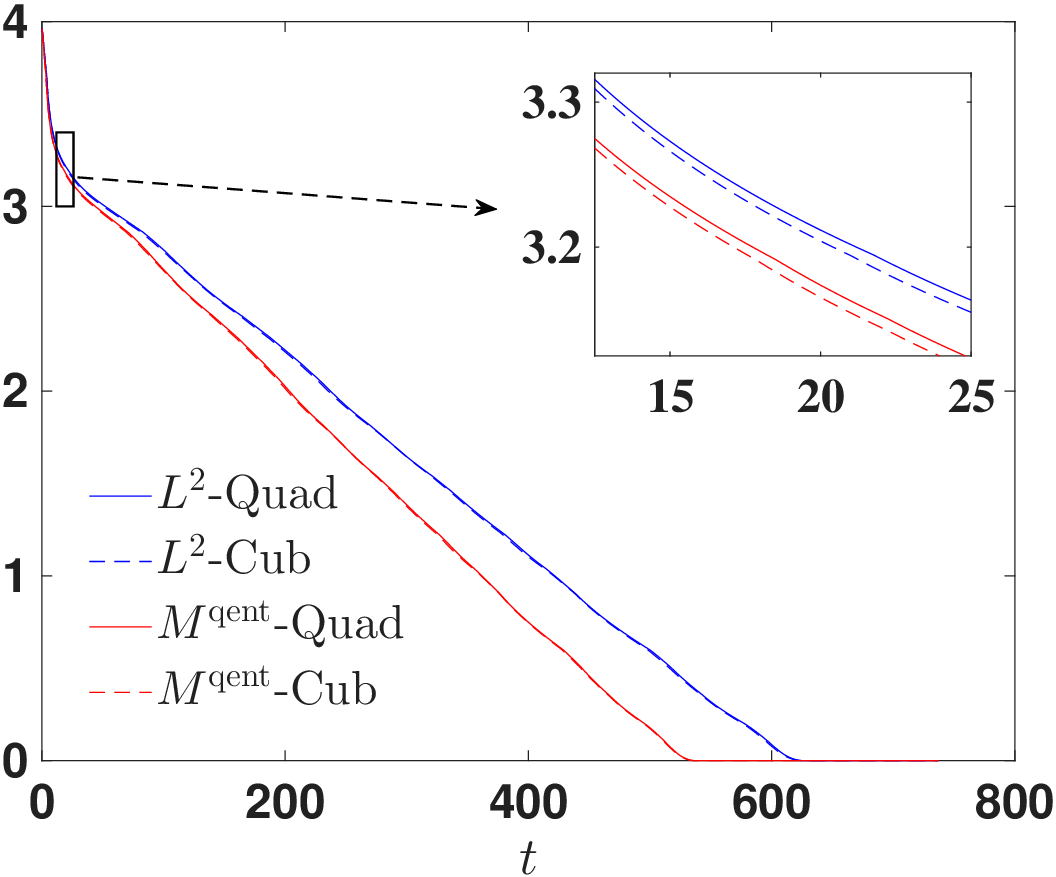}\label{fig:interface 0.12 evolution}}
    \subfigure[Free energy.]{\includegraphics[width=0.315\textwidth]{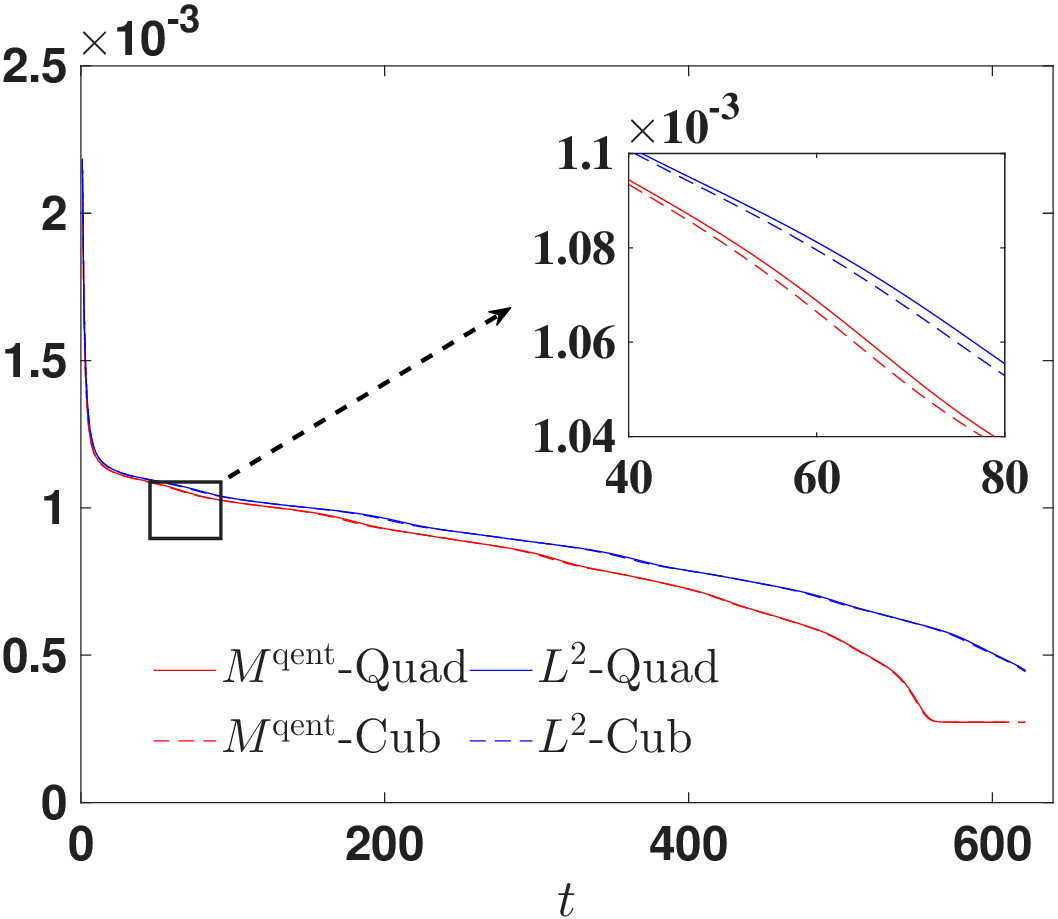}\label{fig:Total energy of interface problem}}
    \subfigure[The maximum and minimum eigenvalues of $\fM^{\mathrm{qent}} $.]{\includegraphics[width=0.31\textwidth]{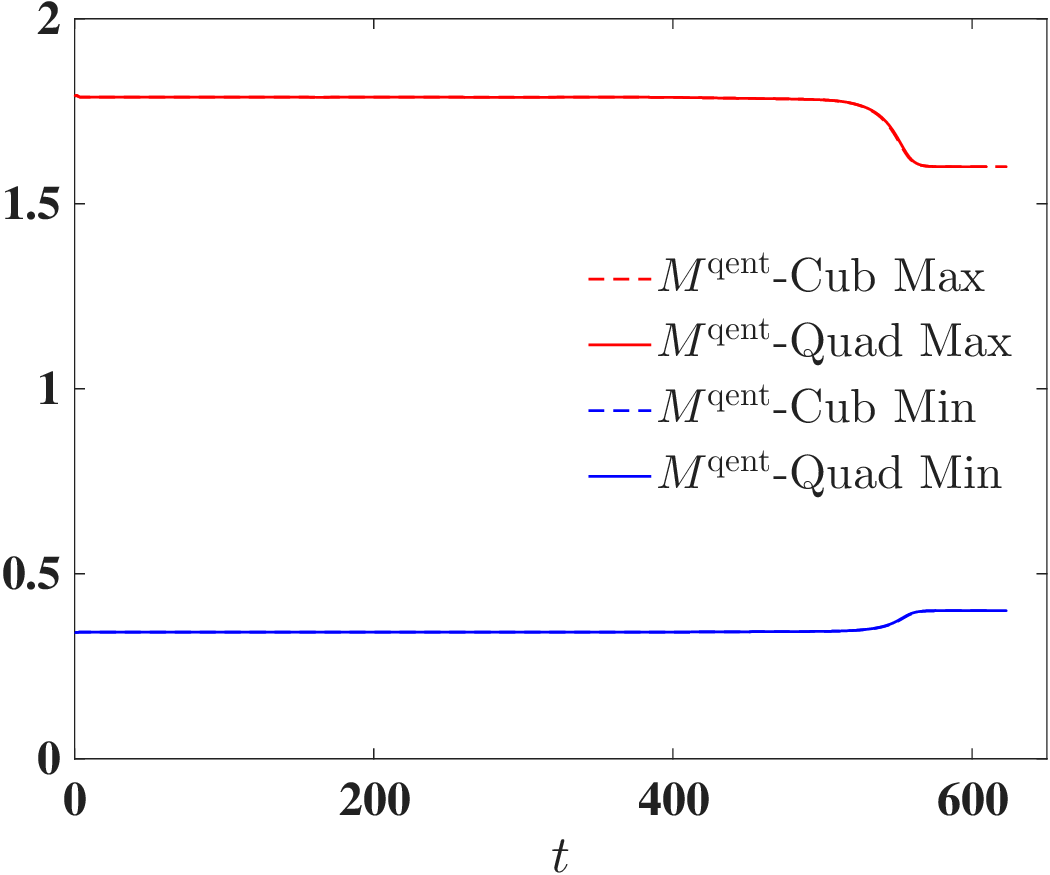}\label{fig:Eigenvalue of interface problem}}
    \caption{The comprasion of four cases for the interface evolution: area occupied by the nematic phase, total energy, and the eigenvalues of $\fM ^{\mathrm{qent}}$.}\label{fig:subfigures}
\end{figure}

\textbf{Isotropic--uniaxial nematic interfaces.}
We next examine the evolution of the interface between the isotropic and the uniaxial nematic phases.
To this end, we need to carefully choose the bulk energy coefficient $c_{02}$ such that two phases coexist. The value is $c_{02}\approx 13.1117$, and the corresponding uniaxial nematic solution gives $s_0\approx 0.1836$.
The elastic coefficients are chosen as 
$c_{21}=1.6\times 10^{-3}$,
$c_{22}=2.0\times 10^{-3}$,
$c_{24}=1.5\times 10^{-3}$.
The initial condition represents a sharp square-shaped isotropic--nematic interface (with a rotation of $\pi/4$), given by 
\begin{equation}\label{eq:interface-Q}
    Q(\bx,0)=s(\bx)
    \left(\e_1^2-\frac{\bbi}{3}\right),     s(\bx)=\left \{
    \begin{array}{ll}
        s_0, & |x-\pi|+|y-\pi|\le 1.5, \\
        0,   & \text{otherwise}.
    \end{array}\right.
\end{equation}
The preferred direction of the uniaxial nematic phase is the $x$-direction $\e_1$. 

We compute the four cases (with and without $\fM$, zero and nonzero $c_{24}$) with $\delta t=5\times 10^{-2}$.
We observe that in all four cases, the interface gradually shrinks and eventually arrive at a globally isotropic state.
The stages of evolution are similar for the four cases, which we illustrate in
Figure~\ref{fig:The evolution of the interface in Case 1}
using a few snapshots of $\fM ^{\mathrm{qent}}$--Cub.
The colorbar in Figure~\ref{fig:The evolution of the interface in Case 1}
represents the largest eigenvalue of $Q$. 
After the sharp corners of the region become rounded,
the square-shaped region gradually transforms into an ellipse.
The long axis of the ellipse coincides with the direction of the uniaxial nematic phase,
also drawn in Figure~\ref{fig:The evolution of the interface in Case 1}. 
The ellipse region then shrinks and finally vanishes.

On the other hand, the rate of interface evolution is different for the four cases. 
More specifically, it turns out that the presence of $\fM^{\mathrm{qent}}$ greatly accelerates the evolution, while the difference generated by the cubic $c_{24}$ term is not that evident. 
This can be acquired by plotting the area occupied by the nematic phase, defined as the $0.1$-contour of the maximum eigenvalue of $Q$ (Figure~\ref{fig:interface 0.12 evolution}).
The evolution of the free energy (Figure~\ref{fig:Total energy of interface problem}) further reflects evolution rate after the maximum eigenvalue of $Q$ grows less than $0.1$. 
We can see that the energy dissipation is met for the four cases.
After the energy of the two $\fM^{\mathrm{qent}}$ cases has reached that of the uniform isotropic state, for the two $L^2$ cases it still takes some time. 
Meanwhile, the rate of evolution is slightly faster for the two Cub cases. 
To further comprehend the effects caused by $\fM^{\mathrm{qent}}$, we plot the eigenvalues of $\Psi_4(\fM^{\mathrm{qent}})$ in Figure~\ref{fig:Eigenvalue of interface problem}.
The acceleration shall be closely related to the fact that the maximum eigenvalue is roughly five times of the minimum eigenvalue. 

\begin{figure}[htbp]
    \centering
    \subfigure[$t=1\times10^{-3}$.]{\includegraphics[width=0.25\textwidth]{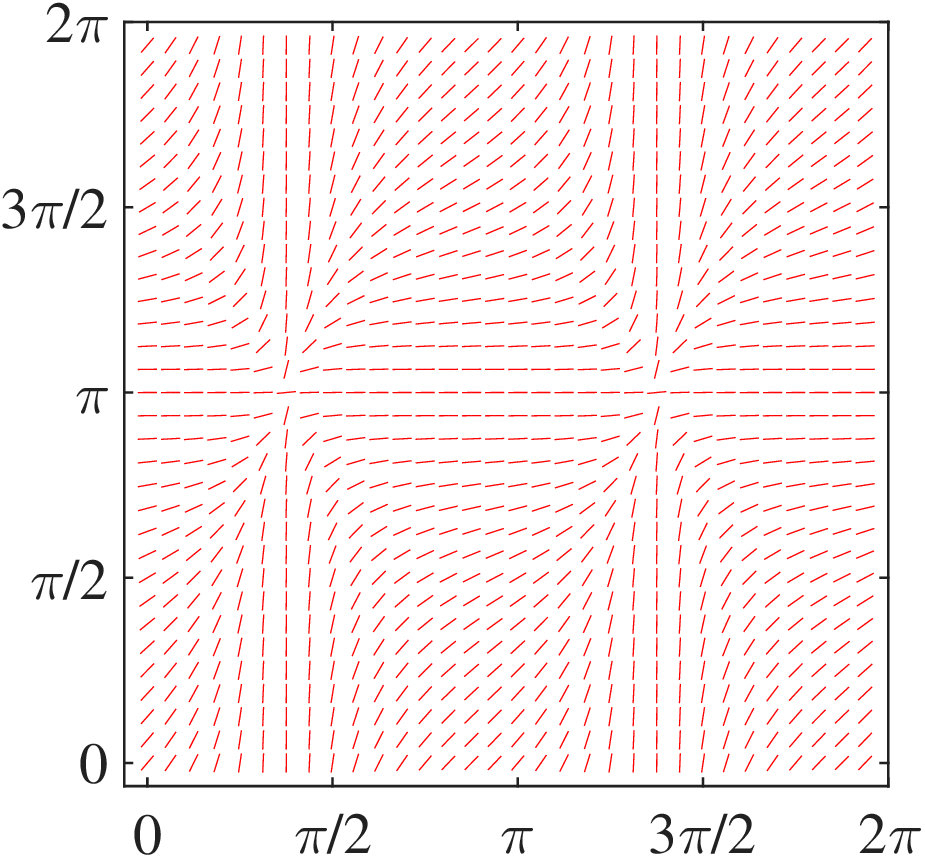}}\hfill
    \subfigure[$t=1\times10^{-2}$.]{\includegraphics[width=0.25\textwidth]{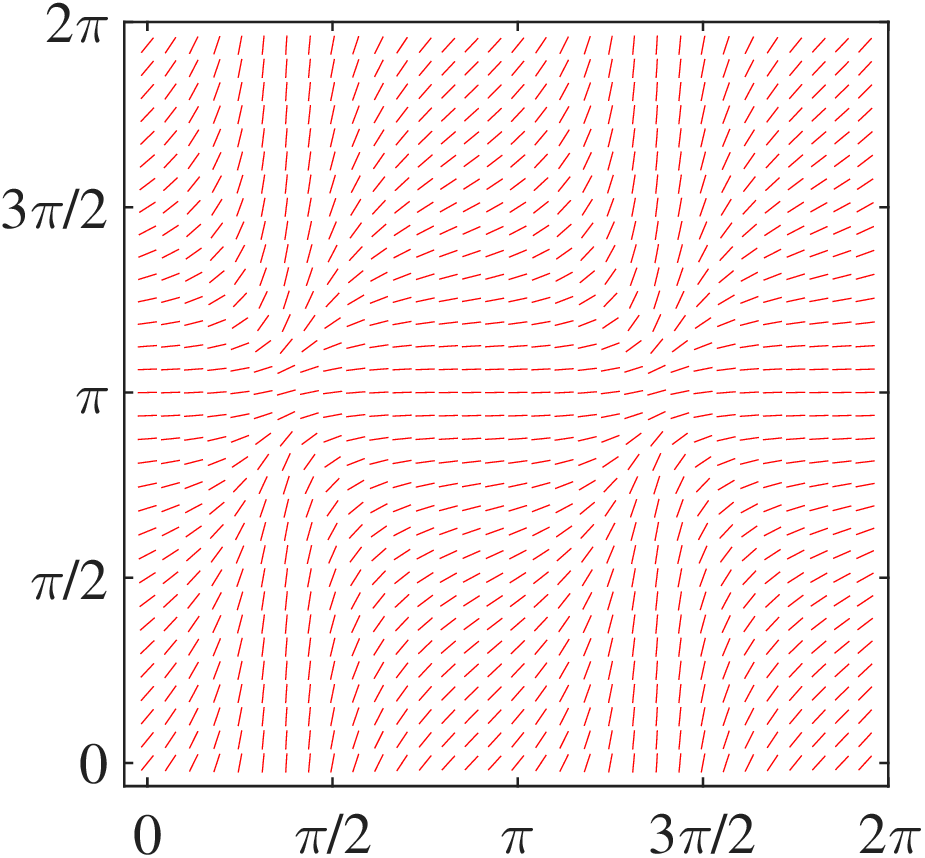}}\hfill
    \subfigure[$t=2\times10^{-2}$.]{\includegraphics[width=0.25\textwidth]{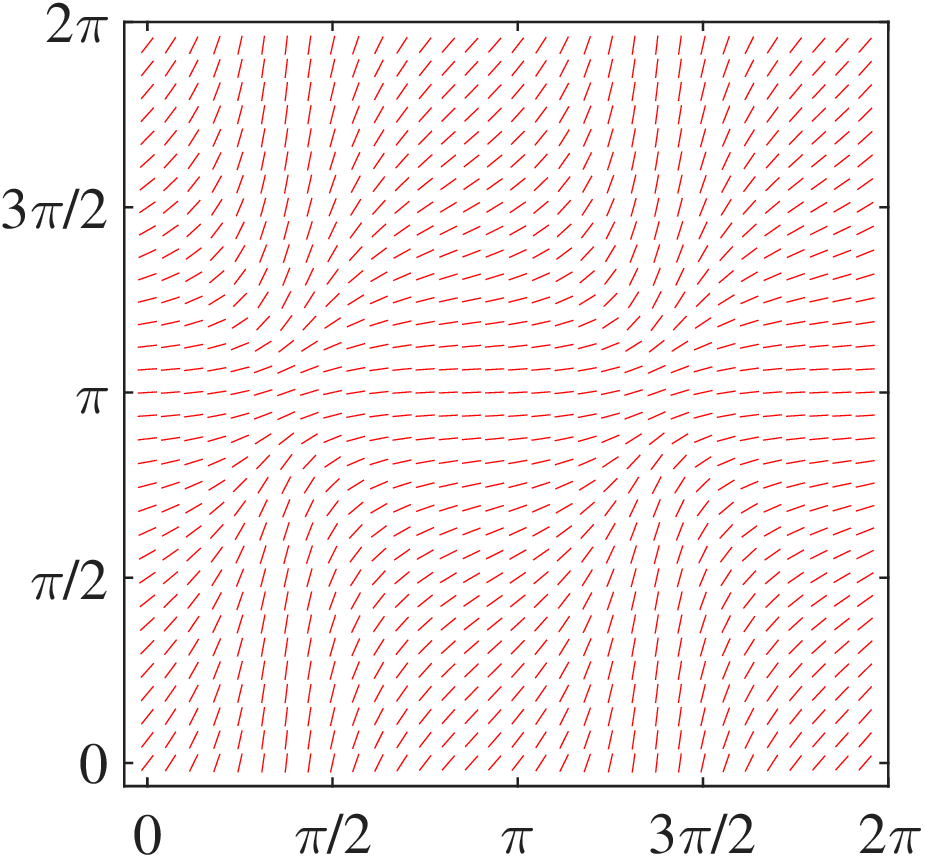}}\hfill
    \subfigure[$t=5\times10^{-2}$.]{\includegraphics[width=0.25\textwidth]{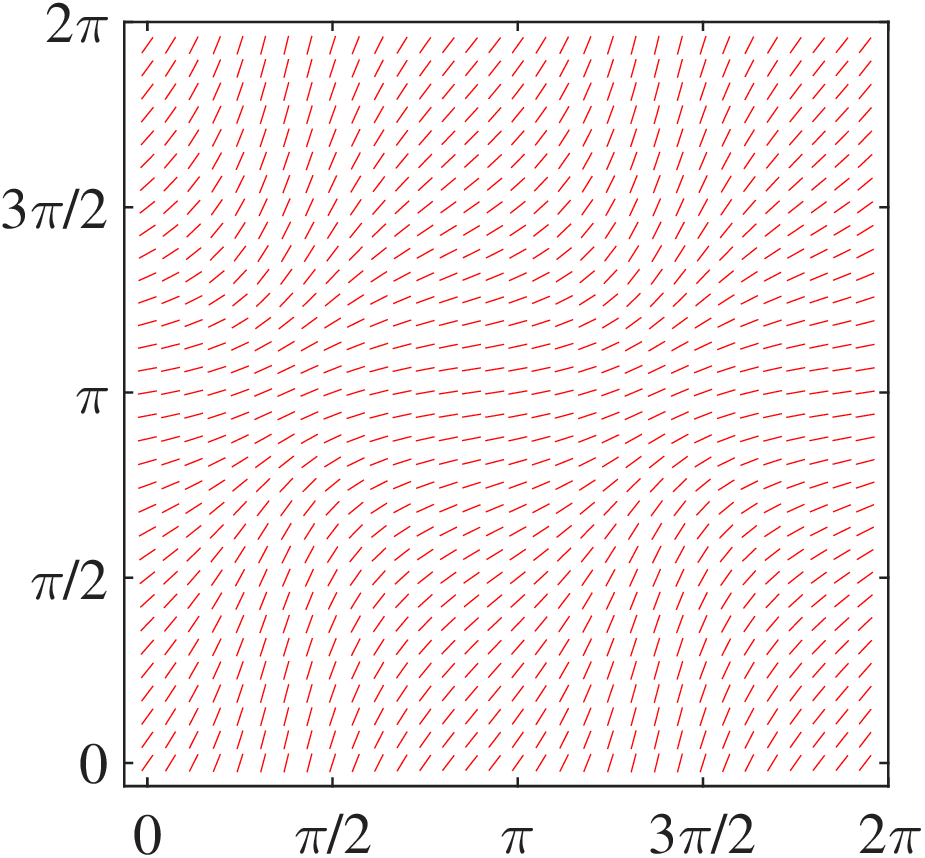}}\hfill
    \caption{The principal eigenvector of $Q$ in $\fM ^{\mathrm{qent}}$--Cub at different time.}\label{fig:The principal eigenvector for Case 1 at different time}
\end{figure}
\begin{figure}[htbp]
    \centering
    \includegraphics[width=1\textwidth]{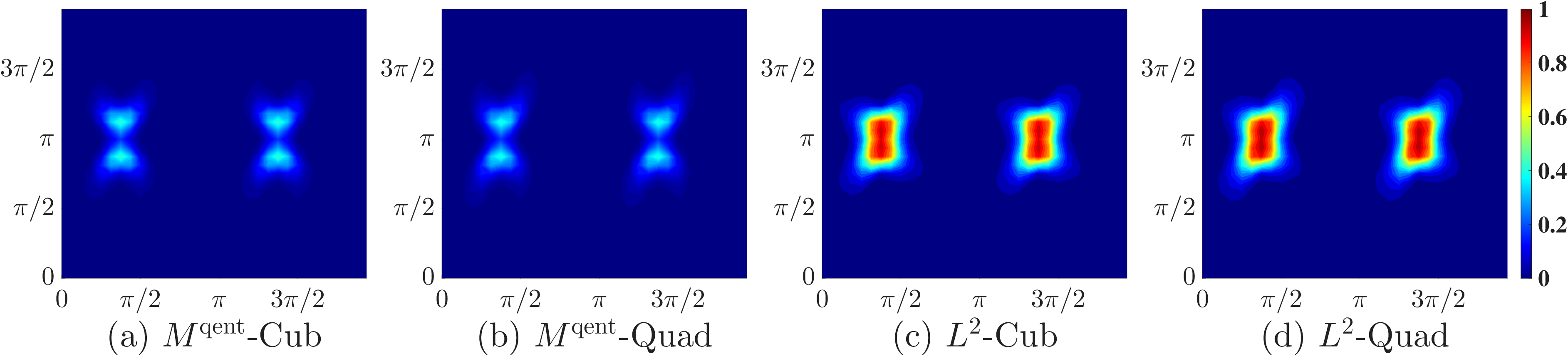}
    \caption{The biaxiality for four cases at $t = 0.02$.}\label{fig:biaxiality dd time}
\end{figure}
\begin{figure}[htbp]
    \centering
    \subfigure[The area enclosed by the $0.1$ contour of the biaxiality.]{\includegraphics[width=0.32\textwidth]{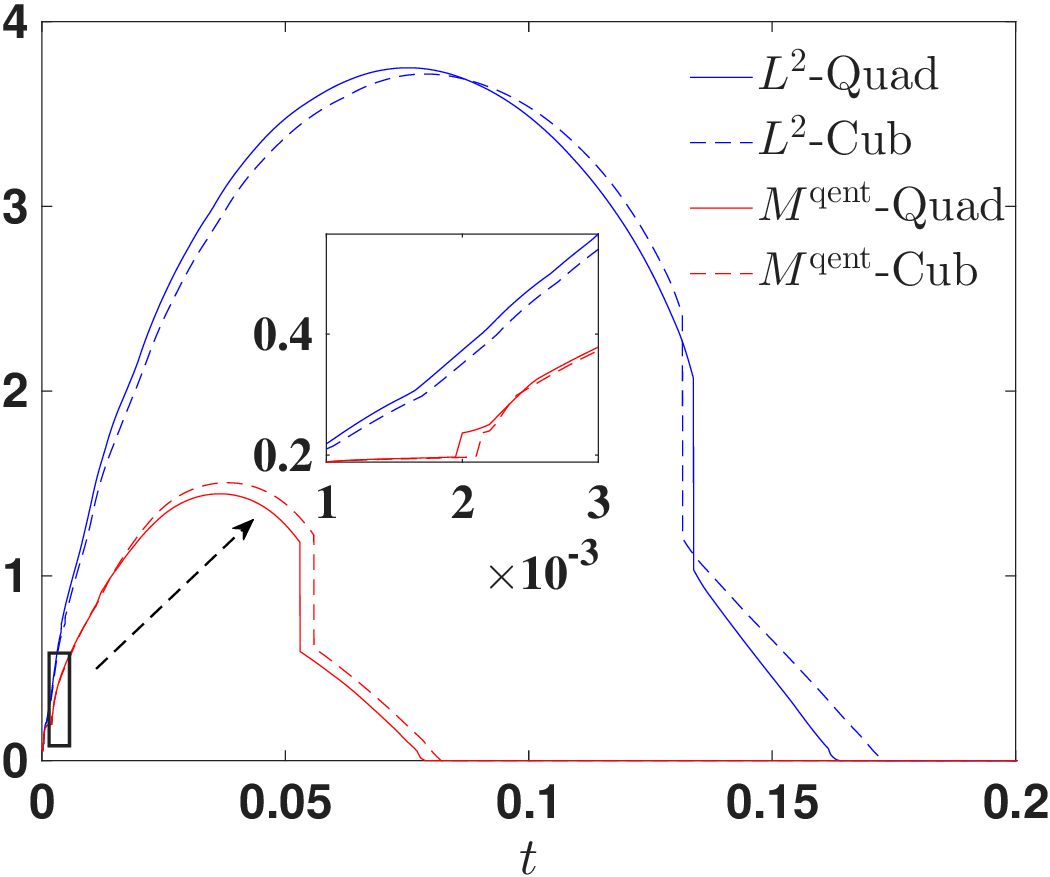}\label{bi_area}}\hfill
    \subfigure[Free energy.]{\includegraphics[width=0.33\textwidth]{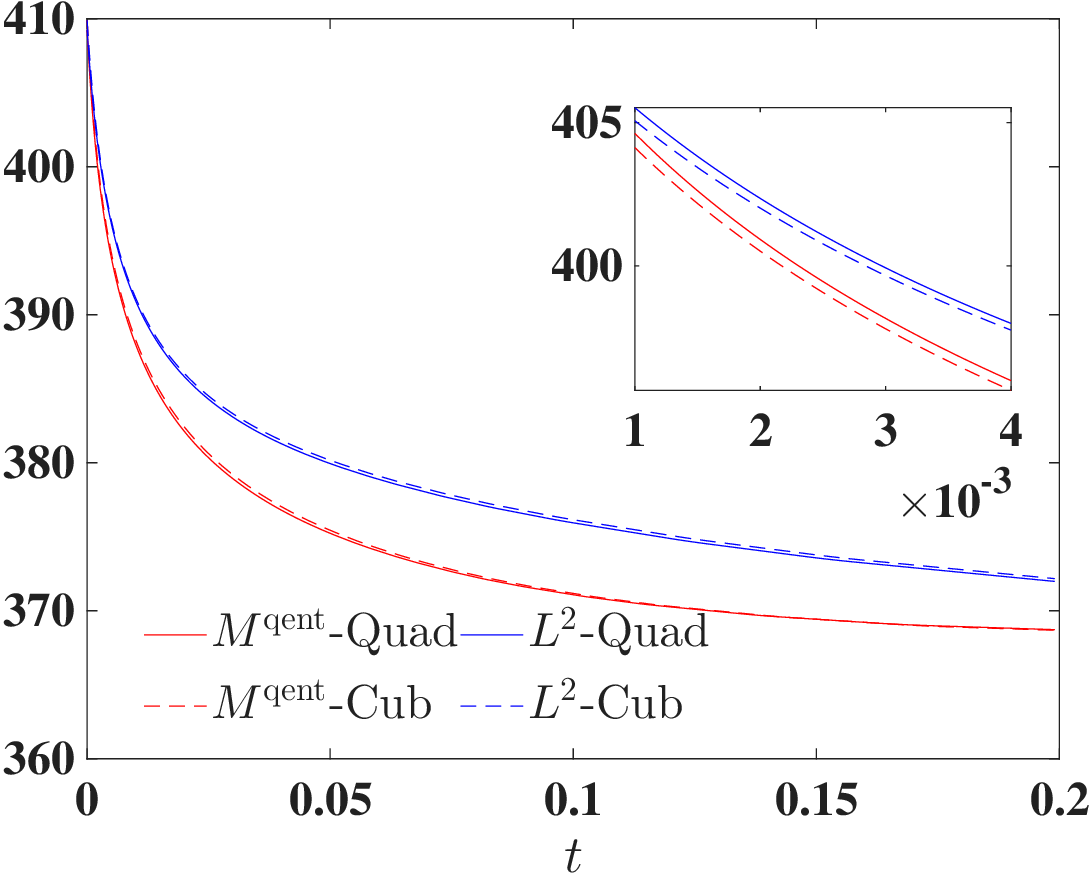}\label{defect_energy}}\hfill
    \subfigure[Numbers of Netwon iteration.]{\includegraphics[width=0.33\textwidth]{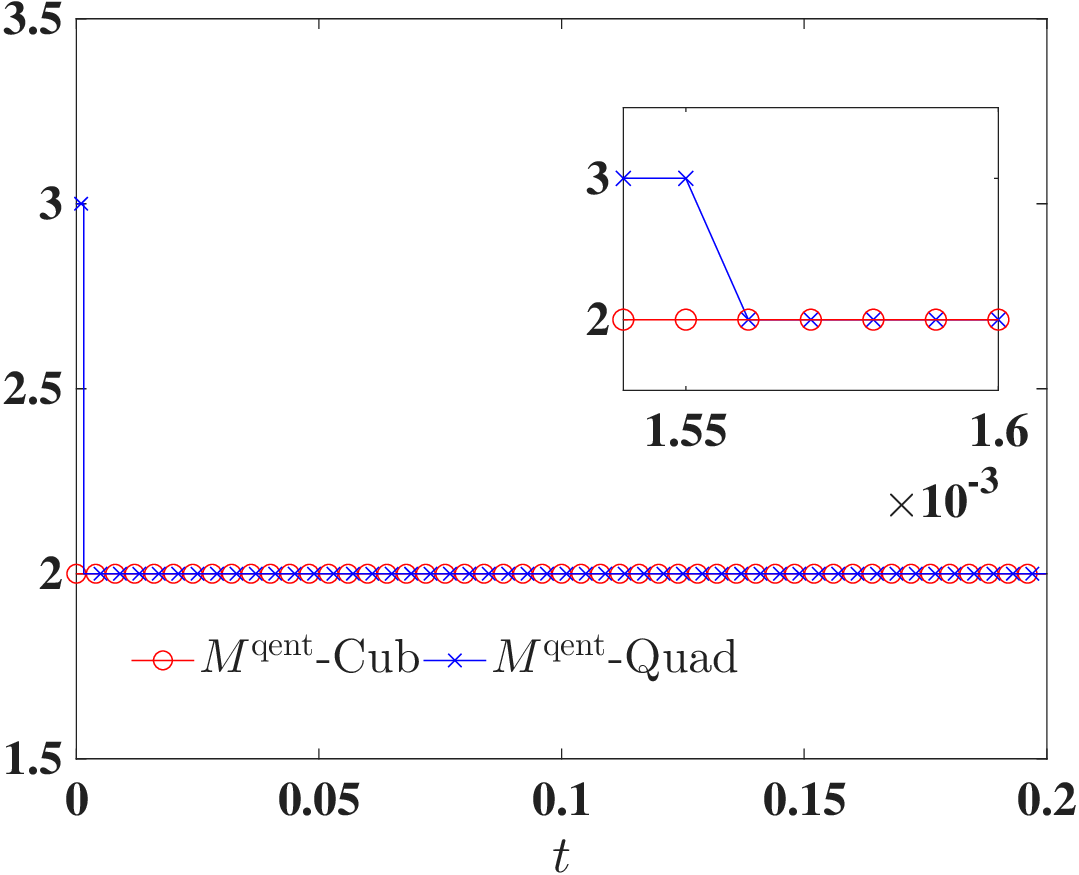}\label{defect_itr}}
    \caption{The comprasion of different cases for the defect evolution: the biaxial area, the free energy and the numbers of Netwon iteration.}\label{fig:the evolution of theta}
\end{figure}

\textbf{Defect evolutions.}
Set 
\[ c_{02} = 20,\,c_{21} = 1.6,\,c_{22} = 2.0,\,c_{24} = 1.5,\,s_0=0.63.\]
Let
\begin{equation*}
  \bu(x,y)=\Big(\cos\big(2(x-\frac{3}{16}\pi)\big)-1,\cos(y-\pi)-1,0\Big)^t,\quad
  \n(x,y)=\frac{\bu}{|\bu|}.
\end{equation*}
It generates discontinuities of \(\n(\bx)\) at two points $({3\pi}/{16},\pi)$ and $({19\pi}/{16},\pi)$.
We choose an initial value $Q_0$ such that it possesses the principal eigenvector $\n(x,y)$ and takes zero at two points where $\n$ is discontinuous, 
\begin{equation*}
    Q_0(x,y)=s_0 \big(1-\exp(-10|\bu|)\big)(\n^2-\frac{\bbi}{3}).
\end{equation*}
The computation is carried out with $\delta t=10^{-5}$.

We first examine the principal eigenvector of $Q$. 
In the four cases (with and without $\fM$, zero and nonzero $c_{24}$), we find that the evolution of the principal eigenvector largely goes through the same stages.
As an example, the $\fM ^{\mathrm{qent}}$--Cub case is drawn in Figure~\ref{fig:The principal eigenvector for Case 1 at different time}, illustrating the disengagement of two discontinuities. 

Despite the similarities for the four cases given by the evolution stages of the principal eigenvector, there are significant differences between the configurations.
To reveal the differences, we investigate the biaxiality quantified by $1-6{{\left(\tr(Q^3)\right)}^2}/{{\left(\tr(Q^2)\right)}^3}\in [0,1]$ for nonzero $Q$.
According to the chosen initial value, the biaxiality is zero at $t=0$ and emerges with $t$ increasing. 
At $t=0.02$, the biaxiality for the four cases is presented in Figure~\ref{fig:biaxiality dd time}, where it is easy to notice that the two $L^2$ cases have greater biaxiality than the two $\fM^{\mathrm{qent}}$ cases. 
This difference actually does not come from the different rates of defect evolution. 
Indeed, we plot the area enclosed by the $0.1$-contour as a function of time $t$ (Figure~\ref{bi_area}), and find that the area is significantly smaller for the two $\fM^{\mathrm{qent}}$ cases.
In other words, one major effect of $\fM^{\mathrm{qent}}$ is that it leads to smaller biaxial regions during the defect evolution.

The free energy evolution is plotted in Figure~\ref{defect_energy}, which still decreases with time. 
Combined with Figure~\ref{bi_area}, it is clear that the disengagements of defects are much faster for the two $\fM^{\mathrm{qent}}$ cases.
In contrast, whether to include a cubic $c_{24}$ term in the free energy appears not to make a big difference. 
In addition, since the numerical scheme is nonlinear, we also look into the number of nonlinear iterations for Newton's method (Figure~\ref{defect_itr}).
For most time steps, only two iterations are needed, thus nonlinearity does not seem to affect the efficiency. 

\section{Concluding Remarks}\label{concl}
For the tensor dynamics derived from the molecular models, we propose the quasi-entropy closure approximation, which is an elementary function proposed for the entropy term in the free energy. 
It is done by reformulating the Bingham closure as a minimization problem of the original entropy function, followed by substituting the original entropy with the quasi-entropy. 
The quasi-entropy closure approximation has the same symmetry properties as the Bingham closure. 
Together with the quasi-entropy in the free energy, in the low velocity approximation we construct a tensor gradient flow with the quasi-entropy. 
Such a tensor model maintains the gradient flow structure of the molecular model and the tensor model with the Bingham closure.
In particular, the $Q$ is constrained within $\Qset$, the dissipation operator given by the higher-order tensor is guaranteed to be positive definite, and the elastic energy may possess a cubic term that is bounded  from below as a result of $Q\in\Qset$. 

The quasi-entropy closure approximation can be done by minimizing an elementary function w.r.t. three variables $a_1$, $a_2$, $a_3$. 
As a result, we write down a first-order-in-time scheme preserving the eigenvalue constraints and energy dissipation, which can be implemented easily from the fact that the discretization of the closure approximation is explicit and thus decoupled from the scheme. 
We examine the evolution of interface and defects, finding that the fourth-order tensor would significantly affect the dynamical behaviors.

When the velocity is not discarded, the tensor dynamics forms a coupled system with the Navier--Stokes equations, where higher-order tensors play different roles.
The quasi-entropy closure approximation can also be incorporated in this system, for which it requires to study whether the essential structures are maintained, and how to construct efficient and accurate numerical methods. 
After that, it would be available to systematically carry out numerical simulations and make meaningful comparisons with previous models. 
The quasi-entropy closure approximation can also be extended to other rigid molecules without axisymmetry.
In this case, more order parameter tensors are included \cite{Xu_2020,xu_2022_kernel,xu_2020_general,Xu_ye_zhang_2018,Xu_Zhang_2017}, leading to many more higher-order tensors to be handled by the closure approximation. 
We expect to investigate these problems in the near future. 

\section*{Acknowledgements}
This work is partially supported by Beijing Natural Science Foundation (No. JQ25002), National Natural Science Foundation of China (Nos. 12288201, 12371414, 12571412, 12171041), National Key R\&D Program of China (No. 2023YFA1008802), the Strategic Priority Research Program of the Chinese Academy of Sciences (No. XDB0510201). 

\bibliographystyle{siamplain}
\bibliography{bibfile}

\end{document}